\documentclass[journal]{IEEEtran}
\usepackage[utf8]{inputenc}

\usepackage{xcolor}
\usepackage{color, soul}
\usepackage{enumerate}
\usepackage[shortlabels]{enumitem}
\usepackage{url}
\usepackage{cite}
\usepackage{tabu}
\usepackage{bm, bbm} 
\usepackage{mathtools}
\usepackage{amsmath, amssymb, amsthm}
\usepackage{nicefrac}
\usepackage{empheq}

\usepackage{booktabs}
\usepackage{graphicx}
\usepackage{multirow}
\usepackage{units}
\usepackage{siunitx}
\usepackage{todonotes}
\usepackage{enumitem}
\usepackage{svg}

\usepackage[noabbrev,capitalize]{cleveref}

\crefname{equation}{}{}
\Crefname{equation}{}{}

\theoremstyle{definition} 
\newtheorem{prop}{Proposition}
\theoremstyle{definition} 
\newtheorem{theorem}{Theorem}
\theoremstyle{remark} 
\newtheorem{remark}{Remark}

\newenvironment{ldescription}[1]
  {\begin{list}{}%
   {\renewcommand\makelabel[1]{##1\hfill}%
   \settowidth\labelwidth{\makelabel{#1}}%
   \setlength\leftmargin{\labelwidth}
   \addtolength\leftmargin{\labelsep}}}
  {\end{list}}

\newcommand{\set}[1]{\mathcal{#1}} 
\newcommand{\norm}[1]{\left\lVert#1\right\rVert} 

\newcommand{\om}{\bm{\omega}}
\newcommand{\Om}{\bm{\Omega}}

\DeclareMathOperator{\Ploss}{Ploss}
\DeclareMathOperator{\Qloss}{Qloss}
\DeclareMathOperator{\LP}{LP}
\DeclareMathOperator{\LQ}{LQ}
\DeclareMathOperator{\Var}{Var}

\DeclareMathOperator{\Stdv}{Stdv}

\DeclareMathOperator{\diag}{diag}

\DeclareMathOperator{\Eptn}{\mathbb{E}}
\DeclareMathOperator{\Prb}{\mathbb{P}}

\usepackage{amsmath,amsfonts,amsthm}

\title{Distribution Electricity Pricing under Uncertainty}
\author{Robert Mieth, \textit{Student Member, IEEE}, and Yury Dvorkin, \textit{Member, IEEE}
\thanks{This work was supported by the NSF Award CMMI-1825212 and ECCS-1847285  and the Alfred P. Sloan Foundation grant G-2019-12363.}
\thanks{R, Mieth and Y. Dvorkin are with the Department of Electrical and Computer Engineering, Tandon School of Engineering, New York University, New York, NY 10012 USA (e-mail: robert.mieth@nyu.edu, dvorkin@nyu.edu)}
\thanks{R. Mieth is also with the Fakultät IV Elektrotechnik und Informatik, Technische Universität Berlin, 10587 Berlin, Germany.}
}

\begin{document}

\pagestyle{empty}
\bstctlcite{IEEE:BSTcontrol} 

\maketitle

\begin{abstract}
 Distribution locational marginal prices (DLMPs) facilitate the efficient operation of low-voltage electric power distribution systems. We propose an approach to internalize the stochasticity of renewable distributed energy resources (DERs) and risk tolerance of the distribution system operator in DLMP computations. This is achieved by means of applying conic duality to a chance-constrained AC optimal power flow. We show that the resulting DLMPs consist of the terms that allow to itemize the prices for the active and reactive power production, balancing regulation, network power losses, and voltage support provided. Finally, we prove the proposed DLMP constitute a competitive equilibrium, which can be leveraged for designing a distribution electricity market, and show that imposing chance constraints on voltage limits distorts the equilibrium. 
\end{abstract}


\section*{Nomenclature}

\noindent \emph{Sets:}
\vspace{-1mm}
\begin{ldescription}{$xxxx$}
\item [$\set{A}_i$] Set of ancestor nodes of node $i$
\item [$\set{C}_i$] Set of children nodes of node $i$
\item [$\set{D}_i$] Set of downstream nodes of node $i$ (including $i$)
\item [$\set{G}$] Set of controllable DERs, \mbox{$\set{G}\subseteq\set{N}$}
\item [$\set{L}$] Set of edges/lines indexed by $i\in\set{N}^+$
\item [$\set{N}$] Set of nodes, indexed by $i =\{0,1,\ldots, n \}$
\item [$\set{N}^+$] Set of nodes without the root node, i.e. $\set{N}^+ =\{\set{N} \backslash 0 \}$
\end{ldescription}

\noindent
\emph{Variables and Parameters:}
\begin{ldescription}{$xxxxx$}
\item [$a_{1-3,c}$] Polygonal approximation coefficients, $c=1,...,12$
\item [$a_i$] Generator cost function parameter (first-order) 
\item [$b_i$] Generator cost function parameter (second-order)
\item [$c_{0-2,i}$] Generator cost function parameters (standard form)
\item [$d_i^{P}$] Active net power demand at node $i$
\item [$d_i^{Q}$] Reactive net power demand at node $i$
\item [$e$] Vector of ones of appropriate dimensions
\item [$f_i^{P}$] Active power flow on edge $i$
\item [$\bar{f}_i^{P}$] Active power flow at a forecasted operating point
\item [$f_i^{Q}$] Reactive power flow on edge $i$
\item [$\bar{f}_i^{P}$] Reactive power flow at a forecasted operating point
\item [$g_i^{P}$] Active power production at node $i$
\item [$g_i^{Q}$] Reactive power production at node $i$
\item [$g_i^{P,\max}$] Maximum active power production of generator $i$
\item [$g_i^{P,\min}$] Minimum active power production of generator $i$
\item [$g_i^{Q,\max}$] Maximum reactive power production of generator $i$
\item [$g_i^{Q,\min}$] Minimum reactive power production of generator $i$
\item [$l_i$] Squared current magnitude on edge $i$
\item [$\bar{l}_i$] Squared current at a forecasted operating point
\item [$r_i$] Resistance of edge $i$
\item [$s^2$] Total sum over the covariance matrix ($s^2= e\Sigma e^{\!\top}$)
\item [$t^v_i$] Auxiliary variable for voltage standard deviation
\item [$t^f_i$] Auxiliary variable for power flow standard deviation
\item [$u_i$] Squared voltage magnitude at node $i$ ($u_i = v_i^2$)
\item [$\bar{u}_i$] Squared voltage at a forecasted operating point
\item [$u_i^{\max}$] Upper voltage limit at node $i$
\item [$u_i^{\min}$] Lower voltage limit at node $i$
\item [$v_i$] Voltage magnitude at node $i$
\item [$x_i$] Reactance of edge $i$
\item [$z_\epsilon$] $(1-\epsilon)$-quantile of the cumulative distribution function

\item [$A$] Flow sensitivity matrix
\item [$\check{A}$] Inverse of $A$ ($\check{A}=A^{-1}$)
\item [$A^L$] Loss-aware extension of $A$
\item [$\LP_{ij}^{P}$] Relative active power losses allocated to node $i$ by changes in active power injection at node $j$
\item [$\LP_{ij}^{Q}$] Relative active power losses allocated to node $i$ by changes in reactive power injection at node $j$
\item [$\LQ_{ij}^{P}$] Relative reactive power losses allocated to node $i$ by changes in active power injection at node $j$
\item [$\LQ_{ij}^{Q}$] Relative reactive power losses allocated to node $i$ by changes in reactive power injection at node $j$
\item [$R$] Matrix mapping nodal active power injections into voltage magnitudes
\item [$\check{R}$] Inverse of $R$ ($\check{R}=R^{-1}$)
\item [$R^L$] Loss-aware extension of $R$
\item [$\check{R}^L$]Inverse of $R^L$ ($\check{R}^L=(R^L)^{-1}$)
\item [$S_i^{\max}$] Apparent power flow limit of edge $i$
\item [$X$] Matrix mapping reactive power injections to voltage magnitudes

\item [$\alpha_i$] Balancing participation factor of generation at node $i$
\item [$\gamma$] Price for balancing regulation
\item [$\epsilon_v$] Probability of voltage constraint violations 
\item [$\epsilon_g$] Probability of generation constraint violations 
\item [$\epsilon_f$] Probability of flow limit constraint violations 
\item [$\lambda^P_i$] DLMP for active power at node $i$
\item [$\lambda^Q_i$] DLMP for reactive power at node $i$
\item [$\rho^f_i$] Auxiliary variable for flow chance constraints 
\item [$\rho^v_i$] Auxiliary variable for voltage chance constraints 
\item [$\sigma^2_i$] Variance at node $i$
\item [$\om$] Vector of random nodal forecast errors ${[\om_i, i\in\set{N}]}$
\item [$\Sigma$] Covariance matrix of $\om$ ($\Sigma =  \Var[\om]$)
\item [$\Om$] Sum of random nodal forecast errors ($\Om=e^{\!\top}\om$)
\end{ldescription}

\section{Introduction}

\IEEEPARstart{N}{odal} electricity pricing has been shown to support the efficient scheduling and dispatch of energy resources at the transmission (wholesale) level \cite{schweppe2013spot}. However, the proliferation of distributed energy resources (DERs) in low-voltage distribution systems and the subsequent growth of independent, small-scale energy producers has weakened a correlation between wholesale electricity prices and distribution electricity rates (tariffs), thus distorting economic signals experienced by end-users \cite{verzijlbergh2014renewable}. To overcome these distortions, distribution locational marginal prices (DLMPs) have been proposed to incentivize optimal  operation and DER investments in low-voltage distribution systems, \cite{sotkiewicz2006nodal,heydt2012pricing,bai2018distribution,li2014distribution,huang2015distribution,zhao2019congestion}, and to facilitate the coordination between the transmission and distribution systems, \cite{papavasiliou2018analysis,yuan2018distribution,ding2013real,caramanis2016co}. 
However, implementing DLMPs in practice is obstructed by the inability to accurately capture stochasticity of renewable generation resources (e.g. solar or wind) in the price formation process, \cite{wong2007pricing,morales2012pricing,martin2015stochastic,kazempour2018stochastic,dall2017chance,bienstock2014chance,dvorkin2019chance,mieth2018online}. 
\textcolor{black}{As a result, prospective distribution market designs lack completeness, i.e. do not offer customized financial instruments to deal with each source of uncertainty, which may result in market inefficiencies and welfare losses, \cite{doi:10.1137/110851778, Philpott2016}. 
Motivated by the need to complete distribution market designs with uncertainty and risk information on renewable generation resources, this paper proposes a new approach to obtain DLMPs that explicitly incorporate the stochasticity of renewable generation resources and analyzes the effect of risk and uncertainty parameters on the price formation process.} 

Previously, DLMPs have been considered for numerous applications.
Similarly to wholesale markets, \cite{li2014distribution,huang2015distribution} propose a distribution day-ahead market to alleviate congestion caused by electric vehicle charging using a welfare-maximizing DC optimal power flow (OPF) model for DLMP computations. 
Alternatively, the model in \cite{sotkiewicz2006nodal} introduces  power losses in DLMP computations to properly reward DERs for reducing system-wide power losses. 
In \cite{heydt2012pricing}, the authors compute energy,  congestion, and power loss DLMP components in the presence of advanced smart grid devices, e.g. solid state transformers and variable impedance lines. 
Furthermore, DLMPs have been shown to support the system operation, e.g. by incentivizing voltage support from DERs, \cite{bai2018distribution}, or by mitigating voltage imbalance in a three-phase system, \cite{zhao2019congestion}.
Papavasiliou \cite{papavasiliou2018analysis} comprehensively analyzes DLMPs and their properties using the branch AC power flow model and its convex second-order conic (SOC) relaxation.
The branch power flow model facilitates the use of spot electricity pricing to analyze the effect of the root node prices, power losses, voltage constraints and thermal line limits on DLMP computations, {\color{black}but yields significant computational complexity} even for small networks. On the other hand, its SOC relaxation makes it possible to represent DLMPs in terms of local information, i.e. parameters of a given distribution node and its neighbors.  However, all  DLMP computations in \cite{li2014distribution,huang2015distribution,sotkiewicz2006nodal, heydt2012pricing, bai2018distribution, zhao2019congestion, papavasiliou2018analysis} disregard stochasticity of renewable DER technologies and, therefore, the resulting prices do not provide proper incentives to efficiently cope with balancing regulation needs. The need to consider stochasticity of renewable generation resources in the price formation process is recognized for transmission (wholesale) electricity pricing, e.g.  \cite{wong2007pricing,martin2015stochastic,morales2012pricing,kazempour2018stochastic, ye2017uncertainty, kuang2018pricing}, but there is no {\color{black}framework for stochasticity-aware pricing in emerging distribution markets.} 

Recently, chance-constrained (CC) programming has been leveraged to deal with stochasticity of DERs in the distribution system and to robustify operating decisions of the distribution system operator (DSO), \cite{dall2017chance,mieth2018data,hassan2018optimal,baker2019joint,bruninx2017valuing}. 
The models in \cite{dall2017chance,mieth2018data,hassan2018optimal,baker2019joint} improve compliance with distribution system limits at a moderate, if any, increase in operating costs. However, with the  exception of  our prior work in \cite{kuang2018pricing,dvorkin2019chancemarket,mieth2019risk}, their application for electricity pricing has not been considered. This paper fills this gap and derives {\color{black}DLMPs that internalize stochasticity} using the chance-constrained framework.
{\color{black} 
This framework offers some significant advantages over other uncertainty-aware methods such as robust or scenario-based stochastic-optimization.
First, chance constraints internalize continuous probability distributions of uncertain parameters, which are readily available from historical data (e.g. weather data can be obtained from National Weather Service operated by National Oceanic and Atmospheric Agency, \cite{dvorkin2015uncertainty}, and load data can be retrieved from archived forecasts, \cite{mieth2018online}). 
Further, they can accommodate a broad variety of parametric distributions, \cite{dvorkin2015uncertainty,roald2015security}, and attain distributional robustness, \cite{mieth2018data,dvorkin2019chancemarket}. 
Hence, unlike scenario-based stochastic and robust optimization methods, chance constraints do not require discretizing a probability space for scenario sampling or for deriving a finite uncertainty set. 
In the electricity pricing and market design context, avoiding such somewhat arbitrary and nontransparent  input data manipulations can improve acceptance of stochastic markets among market participants, \cite{kazempour2018stochastic}. 
Second, chance-constrained programs can be solved efficiently at scale, \cite{lubin2016robust}, and generally yield less conservative results, \cite{bienstock2014chance}. 
The residual conservatism can be tuned via a confidence interval, which can be related to established system reliability metrics, e.g. loss of load probability (LOLP) or expected energy not served (EENS), \cite{wu2014chance}.
Third, the chance-constrained OPF automatically fulfills all internalized market design considerations, e.g. revenue adequacy, cost recovery and incentive compatibility, for all potential outcomes and does not require scenario-specific adjustments, \cite{dvorkin2019chancemarket}, which cause social welfare losses if scenario-based stochastic optimization methods are used, \cite{kazempour2018stochastic}.
}

{\color{black} 
To take advantage of chance constraints, we formulate a CC AC-OPF model for a distribution system with renewable DERs and obtain its SOC equivalent, similarly to the previous work in \cite{mieth2018data,dall2017chance}. This convex equivalent enables the use of duality theory for the main propositions of this paper: 
\begin{enumerate}[(i)]
    \item to compute DLMPs that internalize the stochasticity of renewable DERs and risk tolerance of the DSO and
    \item to itemize DLMP components related to nodal active and reactive power production and demand, balancing regulation, network power losses and voltage support. 
\end{enumerate}
}
{\color{black} 
From a market  perspective, this paper describes an approach to price generic distribution-level energy and reserve products  ahead of   real-time operations (e.g. daily, hourly or sub-hourly). In that sense, the model and pricing approach presented below are similar to  a centralized market clearing problem in \cite{caramanis2016co}, which minimizes social costs, schedules the available capacity of resources, and derives marginal-cost-based transmission and distribution prices for day-ahead, hour-ahead, or 5-min real-time markets.
However, unlike in \cite{caramanis2016co}, chance constraints endogenously determine both the system reserve requirement and its network-constrained allocation given a desired risk level and uncertainty parameters. 
The resulting stochastic DLMPs capture these requirements and allocations and can be used to establish a co-optimized stochastic distribution market\textcolor{black}{, which attains a  market equilibrium under the assumption that all market participants share similar knowledge about uncertain parameters.}
With increasing participation of DERs in the provision of grid support services, such a  distribution marketplace  will allow for  coordinating grid support services  among transmission and distribution systems. 
For example,  DER aggregators in Germany and Belgium seek to provide regulation services  at both the transmission and distribution level, \cite{nextkraftwerke}. However, only the former provision is administrated on a market basis, which  makes distribution services less attractive for DER aggregators.  Hence, DLMPs that internalize desired risk levels and uncertainty parameters can support the DSO in rolling out  efficient market platforms  to incentivize DER participation in supporting  distribution system operations. 

}

\section{Model Formulation}

In the following, we consider a generic low-voltage distribution system with controllable DERs and uncontrollable (behind-the-meter) stochastic  generation resources. 
{\color{black} 
All controllable DERs are small-scale generators with given production costs and constant generation limits.}
\textcolor{black}{The DSO is responsible for scheduling and dispatching controllable generation resources in the least-cost manner using the best available forecasts, while meeting technical limits on distribution system operations and computing DLMPs.}
The distribution system is a radial network given by graph $\Gamma (\set{N}, \set{L})$, as in Fig.~\ref{fig:notation_illustration}, where $\set{N}$ and $\set{L}$ are the sets of nodes indexed by $i \in \{0,1,...,n\}$ and edges (lines) indexed by $i\in\{1,2,...,n\}$.
The root node, indexed as $0$, is the substation, i.e.  an infinite power source with fixed voltage $v_0$, and $\set{N}^+ \coloneqq \set{N} \setminus \{0\}$ is the set of all non-root nodes. 
Each node is associated with ancestor (or parent) node $\set{A}_i$, a set of children nodes $\set{C}_i$ and a set of downstream  nodes  $\set{D}_i$ (including $i$). 
Since $\Gamma$ is radial, it is \mbox{$|\set{A}_i| = 1, i \in \set{N}^+$} and all edges {\color{black}$i \in \set{L}$} are indexed by $\set{N}^+$. 

Each node is characterized by its active and reactive net power demand  ($d_i^P$ and $d_i^Q$, $i \in \set{N}$),  i.e. the difference between the nodal load and behind-the-meter DER output, and voltage magnitude $v_{i} \in [v_{i}^{\min}, v_{i}^{max}], i \in \set{N}$, where $v_{i}^{\min}$ and $v_{i}^{\max}$ are the upper and lower voltage limits. 
To use linear operators, we introduce $u_{i} = v_{i}^2, i \in \set{N}$ limited by $u_i^{\min} = (v_i^{\min})^2$ and $u_i^{\max} = (v_i^{\max})^2$. 
\textcolor{black}{If node $i$ hosts a controllable DER, the active and reactive power output is modeled as dispatchable with ranges $g_{i}^P \in [g_{i}^{P,\min}, g_{i}^{P,\max}]$ and $g_{i}^Q \in [g_{i}^{Q,\min},g_{i}^{Q,\max}]$, $i \in \set{G} \subseteq \set{N}^+$.}
Active and reactive power flows on edge $i$  with resistance $r_i$, reactance $x_i$ and apparent power limit~$S_i^{\max}$ are given by  $f_i^P$ and $f_i^Q$. Vectors $r = [r_1, \ldots, r_n]$ and $x = [x_1, \ldots, x_n]$ collect all resistances and reactances. 
Symbols in \textbf{bold} font indicate stochastic (random) variables.

\begin{figure}
\centering
\includegraphics[width=0.95\linewidth]{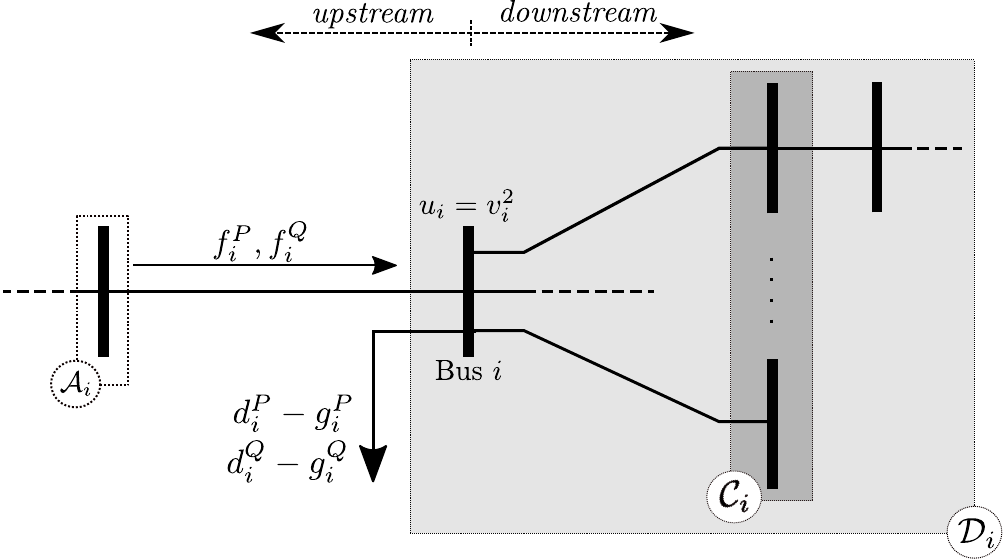}
\caption{Power flow notations in a radial network.}
\label{fig:notation_illustration}
\end{figure}

\subsection{Uncertainty Model and Real-Time Balancing Regulation}

The  active net power demand is modeled as:
\begin{align}
    \bm{d}_i^P = d_i^P + \om_i \label{eq:uncertainty}
\end{align}
where $\bm{d}_i^P$ is  the stochastic demand, $d_i^P$ is a given net demand forecast and $\om_i$ is a given random forecast error. 
We introduce $\Om \coloneqq \sum_i \om_i = e^{\!\top}\om$, where $\om$ is the column vector collecting all nodal forecast errors and $\om_i$ follows a unimodal distribution with a finite variance. The net demand forecast is assumed to be unbiased, i.e. the expected value and covariance matrix of $\om$ are given by $\Eptn[\om] = 0$ and $\Var[\om] = \Sigma$.

Accounting for forecast error $\om$ requires procuring balancing regulation capacity to continuously match the power supply and demand.   The burden of balancing regulation is distributed among controllable DERs  using participation factors, \cite{bienstock2014chance,dall2017chance,lubin2016robust}. Participation factors $0 \leq \alpha_i \leq 1$ are modeled as  decision variables and  represent a relative amount of the system-wide forecast error ($ \Om$) that  the DER at node~$i$ must compensate for. Therefore, the real-time active power output of each controllable DER ($\bm{g}_i^P$) can be modeled as:
\begin{align}
    \bm{g}_i^P = g_i^P + \alpha_i \Om. \label{eq:uncertain_generation}
\end{align}
Thus, the system remains balanced if ${\sum_{i \in \mathcal N} \alpha_i = 1}$.

\subsection{CC AC-OPF Formulation}

Given the uncertainty and balancing regulation models in \eqref{eq:uncertainty} and \eqref{eq:uncertain_generation}, the CC AC-OPF problem for the DSO follows:

\allowdisplaybreaks
\begin{subequations} \label{eq:ccacopf}
\begin{align}
    &\text{CC AC-OPF:} && \min_{\substack{\{g_i^P, g_i^Q,\alpha_i\}_{i\in \set{G}}, \\ \{f_i^P, f_i^Q, u_i\}_{i\in\set{N}^+}}}  \Eptn\left[\sum_{i=0}^n c_i(\bm{g}_i)\right] \label{cc-opf:objective}
\end{align}
\vspace*{-0.5cm}
\begin{align}
    \text{s.t.} \nonumber \\
    &(\lambda^P_0): && g_0^P - \sum_{j\in \set{C}_0}f_j^P = 0 \label{cc-opf:lambda_0}\\
    &(\lambda^Q_0): && g_0^Q - \sum_{j\in \set{C}_0}f_j^Q = 0 \label{cc-opf:pi_0}\\
    &(\lambda^P_i): && f_i^P + g_i^P - \sum_{j\in \set{C}_i}f_j^P = d_i^P && i\in\set{N}^+ \label{cc-opf:lambda_i}\\
    &(\lambda^Q_i): && f_i^Q + g_i^Q - \sum_{j\in \set{C}_i}f_j^Q = d_i^Q && i \in \mathcal N^+\label{cc-opf:pi_i}\\
    &(\beta_i): && u_i + 2(r_i f_i^P + x_i f_i^Q)  =  u_{\set{A}_i} \hspace{-2cm}&& \label{cc-opf:beta_i} 
    i\in \mathcal N^+ \\
    &(\mu_i^+): && \Prb[\bm{u}_i \leq u_i^{\max}] \geq 1-\epsilon_v && i\in\set{N}^+  \label{cc-opf:mu_i+}\\
    &(\mu_i^-): && \Prb[-\bm{u}_i \leq -u_i^{\min}] \geq 1-\epsilon_v && i\in\set{N}^+ \label{cc-opf:mu_i-}\\
    &(\delta_i^+): && \Prb[\bm{g}_i^P \leq g_i^{P,\max}] \geq 1-\epsilon_g && i\in\set{G} \label{cc-opf:delta_i+}\\
    &(\delta_i^-): && \Prb[-\bm{g}_i^P \leq -g_i^{P,\min}] \geq 1-\epsilon_g && i\in\set{G} \label{cc-opf:delta_i-}\\
    &(\eta_i): && \Prb\left[(\bm{f}_i^P)^2 + (f_i^Q)^2 \leq (S_i^{\max})^2\right] \geq 1-\epsilon_f \hspace{-2cm} && \nonumber \label{cc-opf:eta_i} \\
    & && &&i \in \mathcal N^+ \\
    &(\theta_i^+): && g_i^Q \leq g_i^{Q,\max} && i\in\set{G} \label{cc-opf:theta_i+}\\
    &(\theta_i^-): && -g_i^Q \leq -g_i^{Q,\min} && i\in\set{G} \label{cc-opf:theta_i-} \\
    &(\gamma): && \sum_{i=1}^n \alpha_i + \alpha_0 = 1 \label{cc-opf:gamma}
\end{align}%
\label{mod:cc-opf}%
\end{subequations}%
\allowdisplaybreaks[0]%
Eq.~\cref{cc-opf:objective} \textcolor{black}{accounts for the perspective of the risk-neutral DSO} and minimizes the expected operating cost given the cost functions of controllable DERs and their power output under uncertainty. 
Eq.~\cref{cc-opf:lambda_0,cc-opf:pi_0,cc-opf:lambda_i,cc-opf:pi_i} are  linearized AC power flow equations for distribution systems based on the \textit{LinDistFlow} formulation, \cite{baran1989optimal,turitsyn2010local},  that balance the active and reactive power at each node. 
Eq.~\cref{cc-opf:beta_i} computes the voltage magnitude squared at each node. 
Eqs.~\cref{cc-opf:mu_i+,cc-opf:mu_i-,cc-opf:delta_i+,cc-opf:delta_i-,cc-opf:eta_i} ensure that the  voltage magnitudes squared, power flows and generation outputs  under uncertainty will not violate their respective limits with a given confidence level (probability) chosen by the DSO.
{\color{black} 
The confidence levels in \eqref{cc-opf:mu_i+}--\eqref{cc-opf:eta_i}  can be related to commonly used system reliability metrics, e.g. loss of load probability (LOLP) or expected energy not served (EENS), \cite{wu2014chance}, and assist the DSO in trading off the likelihood of constraint violations and the solution conservatism. 
}%
Eq.~\cref{cc-opf:gamma} ensures that the procured balancing regulation capacity is sufficient to match the power supplied and consumed. 
Greek letters in parentheses in \eqref{cc-opf:lambda_0}--\eqref{cc-opf:gamma} denote dual multipliers of the respective constraints. 

\subsection{Deterministic Equivalent of the CC AC-OPF}
\label{ssec:deterministic_equivalent_of_ccacopf}

We recast \cref{mod:cc-opf} as a computationally efficient SOC program using a reformulation from   \cite{bienstock2014chance,lubin2016robust, mieth2018data}.

\subsubsection{Expected generation cost} 

Controllable DERs have the following quadratic cost function:
\begin{align}
    c_i(g_i^P) = c_{2,i} (g_i^P)^2 + c_{1,i} g_i^P + c_{0,i}. \label{eq:quadratic_cost_of_der}
\end{align}
For the compactness of the subsequent formulations, we denote $c_{2,i}=1/2b_i$, $c_{1,i} = a_i/b_i$, $c_{0,i} = a_i^2/2b_i$. Given these notations, the deterministic equivalent of the cost function of each DER  given by \eqref{eq:quadratic_cost_of_der} is as follows  (as derived in Appendix~\ref{ax:expected_generation_cost}):
\begin{align}
    \Eptn[c_i(\bm{g}_i^P)] = \frac{(g_i^P + a_i)^2}{2 b_i} + \frac{\alpha_i^2}{2 b_i} s^2, \label{eq:cost_function}
\end{align}
where $a_i \geq 0$ and $ b_i > 0$ are given parameters and ${s^2 \coloneqq e^{\!\top}\Sigma e}$.
Note that power provision by the substation (at node $0$) is modeled with the same cost function.

\subsubsection{Generation Chance Constraints}
As in \cite{bienstock2014chance,lubin2016robust,dvorkin2019chance}, \cref{cc-opf:mu_i+,cc-opf:mu_i-,cc-opf:delta_i+,cc-opf:delta_i-} are the chance constraints linearly dependent on uncertainty $\om$ that can be reformulated into the SOC constraint of the following form:
\begin{align}
    x + \Phi^{\!-1}(1-\epsilon)\Stdv[\bm{x}(\om,\alpha)]  \leq x^{\max} \label{eq:cc_reform_general}
\end{align}
where  $x \coloneqq \Eptn[\bm{x}]$ and $\Stdv[\bm{x}(\om,\alpha)] = \sqrt{\Var[\bm{x}(\om,\alpha)]}$ are the expectation and standard deviation of random variable~$\bm{x}$.
If $\om$ is normally distributed then $\Phi^{-1}(1-\epsilon)$ is the $(1-\epsilon)$-quantile of the standard normal distribution.
For more general distributions we refer the interested reader to \cite{roald2015security} for the choice of $\Phi$.
Using \cref{eq:cc_reform_general}, we reformulate the chance constraints  in \cref{cc-opf:delta_i+,cc-opf:delta_i-} as given by  \cref{det-cc-opf:delta_i+,det-cc-opf:delta_i-}, where $z_{\epsilon_g}= \Phi^{-1}(1-\epsilon_g)$ and ${s \coloneqq \sqrt{e^{\!\top}\Sigma e}}$.
{\color{black} 
Using the values of $\alpha_i$, $z_{\epsilon_g}$ and $s$, the reserve margins provided by each controllable DER can be computed as $z_{\epsilon_g} s \alpha_i$ and  then the overall system-wide reserve can be computed as $\sum_{i\in\set{G}}z_{\epsilon_g} s \alpha_i$, see \cite{dvorkin2019chance} for a detailed discussion.
}

\subsubsection{Voltage Chance Constraints}
\label{sssec:voltage_chance_constraints}
Following \cite{mieth2018data}, the nodal  voltage magnitude squared under uncertainty  can be expressed in terms of nodal forecast error $\om$, system-wide forecast error $\Om$ and vector $\alpha \in \mathbb{R}^n$ as follows:
\begin{align}
    \bm{u}_i(\om, \alpha) = u_i - 2 \sum_{j=1}^n R_{ij}(\om_j + \alpha_j\Om) \label{eq:voltage_reformulated}
\end{align}
where $\alpha$ is defined such that its $i$-th entry is zero if no controllable DER is connected to node $i$ and $\alpha_i$ if node $i$ hosts a controllable DER, $R_{ij}$ refers to the elements of the matrix defined as:
\begin{align}
    R \coloneqq A^{\!\top} \diag(r) A \label{eq:def_R_X}
\end{align}
and $A\in\{0,1\}^{n\times n}$ such that $A_{ij}=1$, if edge $i$ is part of the path from the root node $0$ to node $j$, and $A_{ij}=0$ otherwise. 
Given $\alpha$, we introduce an auxiliary notation: 
\begin{align}
    \rho_i^v = [R\alpha]_i = R_{i*} \alpha,
    \label{eq:def_rho}
\end{align}
\noindent
where $[\cdot]_i$ denotes the $i$-th entry of the vector expression and $R_{i*}$ is the $i$-th row of $R$.
This makes it possible to express $\Stdv[\bm{u}_i(\om,\alpha)]$ as follows (see Appendix~\ref{ax:voltage_variance} for the full derivation):
\begin{align}
    \Stdv[\bm{u}_i(\om,\alpha)] 
         &= 2 \norm{(R_{i*} + \rho_i^v e^{\!\top})\Sigma^{1/2}}_2. \label{eq:voltage_stdev_conic}
\end{align}
Using the result in \cref{eq:voltage_stdev_conic}, we can leverage the reformulation given by \eqref{eq:cc_reform_general} and replace the chance constraints in  \eqref{cc-opf:mu_i+}--\eqref{cc-opf:mu_i-} with the convex SOC constraints. 
These SOC constraints are given in \cref{det-cc-opf:zeta_i,det-cc-opf:nu_i,det-cc-opf:mu_i+,det-cc-opf:mu_i-}, where $z_{\epsilon_v} = \Phi^{-1}(1-\epsilon_v)$
{\color{black} and $t_i^v$ is an auxiliary variable.}
Note that \cref{det-cc-opf:nu_i} is equivalent to~\cref{eq:def_rho} because $\check{R} \coloneqq R^{-1}$.
Because the elements of matrix $R$ are resistances $r_i$, which are positive and non-zero, $R$ is positive definite and symmetric due to \cref{eq:def_R_X}, i.e. invertible. 

\subsubsection{Apparent power flow chance constraints}
The chance-constrained limit on apparent power flows  \cref{cc-opf:eta_i} can be approximated by an inner polygon to accommodate the quadratic dependency on the uncertain parameter \cite{wang2016distributed,dvorkin2019chance}:
\begin{align}
    a_{1,c} \bm{f}_i^P + a_{2,c} f_i^Q + a_{3,c} S_i^{\max} \leq 0  \quad  c = 1...12
\end{align}
where $a_{1,c}$, $a_{2,c}$ and $a_{3,c}$ are coefficients of the set of the linearized constraints.
Using $A$ as defined in \eqref{eq:def_R_X} and expressing $\rho_i^f = [A\alpha]_i = A_{i*} \alpha$, the standard deviation of the  active power flow is given by 
\begin{align}
    \Stdv[\bm{f}_i^P(\om,\alpha)] 
         &= \norm{(A_{i*} - \rho_i^f e^{\!\top})\Sigma^{1/2}}_2. \label{eq:active_flow_stdev_conic}
\end{align}
Using the result in \eqref{eq:active_flow_stdev_conic} and the standard reformulation in  \eqref{eq:cc_reform_general}, we can replace the chance constraint in \cref{cc-opf:eta_i} with the SOC constraints in  \cref{det-cc-opf:nuf_i,det-cc-opf:zetaf_i,det-cc-opf:eta_i}, where  $\check{A} \coloneqq A^{-1}$ and $z_{\epsilon_f} = \Phi^{-1}(1-\epsilon_f)$ {\color{black} and $t_i^f$ is an auxiliary variable.}

\subsubsection{Deterministic CC AC-OPF Equivalent}
Reformulating the objective function and chance constraints as described above leads to  the following deterministic SOC equivalent of the CC AC-OPF  in \eqref{eq:ccacopf}:

\begin{subequations} \label{eq:eqvCC}
\begin{align}
     &\text{EQV-CC:} && \min_{\substack{\{g_i^P, g_i^Q,\alpha_i\}_{i\in \set{G}}, \\ \{f_i^P, f_i^Q, u_i\}_{i=\set{N}^+}}} 
    \sum_{i=0}^n\left(c_i(g_i^P) + \frac{\alpha_i^2}{2 b_i} s^2\right)
    \label{det-cc-opf:objective}
\end{align}
\begin{align}
   \text{s.t.} & && \text{\cref{cc-opf:lambda_0,cc-opf:lambda_i,cc-opf:pi_0,cc-opf:pi_i,cc-opf:beta_i,cc-opf:theta_i+,cc-opf:theta_i-,cc-opf:gamma}} \nonumber \\
    &(\delta_i^+): && g_i^P + z_{\epsilon_g} s \alpha_i \leq g_i^{P,\max} && i\in\set{G} \label{det-cc-opf:delta_i+}\\
    &(\delta_i^-): && -g_i^P + z_{\epsilon_g} s \alpha_i \leq -g_i^{P,\min} && i\in\set{G} \label{det-cc-opf:delta_i-}\\
    &(\zeta_i): && t_i^v \geq \norm{(R_{i*} + \rho_i^v e^{\!\top})\Sigma^{1/2}}_2 && i\in\set{N}^+ \label{det-cc-opf:zeta_i} \\
    &(\nu_i^v): && \sum_{j=1}^n \check{R}_{ij} \rho_j^v = \alpha_i && i\in\set{G} \label{det-cc-opf:nu_i}\\
    &(\mu_i^+): && u_i + 2 z_{\epsilon_v} t_i^v \leq u_i^{\max} && i\in\set{N}^+ \label{det-cc-opf:mu_i+}\\
    &(\mu_i^-): && -u_i + 2 z_{\epsilon_v} t_i^v \leq -u_i^{\min} && i\in\set{N}^+  \label{det-cc-opf:mu_i-} \\
    &(\nu_i^f): && \sum_{i=1}^n \check{A}_{ij} \rho_j^f = \alpha_i && i\in\set{G} \label{det-cc-opf:nuf_i}\\
    &(\zeta_i^f): && t_i^f \geq \norm{(A_{i*} - \rho_i^f e^{\!\top})\Sigma^{1/2}}_2 && i\in\set{N}^+ \label{det-cc-opf:zetaf_i}\\
    &(\eta_{i,c}): && a_{1,c} (f_i^P + z_{\epsilon_f}  t_i^f) + a_{2,c} f_i^Q + a_{3,c} S_i^{\max} \leq 0 \hspace{-4cm} && \nonumber \\
    & && && \hspace{-1.9cm} i\in\set{N}^+, c \in \{1,...,12\} \label{det-cc-opf:eta_i} 
\end{align} \vspace{-5mm}
\label{mod:EQV-CC}
\end{subequations}

{\color{black} 
\subsection{Extension with Loss Factors}
\label{ssec:extension_with_loss_factors}

The EQV-CC in \eqref{eq:eqvCC} can be extended to incorporate power losses to account for their effect on prices.
For this purpose we derive an approximate linear mapping of nodal net injections into power losses, \cite{yuan2018novel,bai2018distribution,dvorkin2019chance}.
The active and reactive power losses on edge $i$ are given by $l_i r_i$ and $l_i x_i$, where $l_i$ is the squared current on edge $i$.
Using the available demand forecast, we solve:
\allowdisplaybreaks
\begin{subequations} 
\begin{align}
    & && \min_{\substack{\{g_i^P, g_i^Q,\alpha_i\}_{i\in \set{G}}, \\ \{f_i^P, f_i^Q, u_i\}_{i=\set{N}^+}}} 
    \sum_{i=0}^n\left(c_i(g_i^P) + \frac{\alpha_i^2}{2 b_i} s^2\right) \hspace{-3cm} &
    \\
   \text{s.t.} & && \text{\cref{cc-opf:theta_i+,cc-opf:theta_i-,cc-opf:gamma,det-cc-opf:delta_i+,det-cc-opf:delta_i-}} & \nonumber \\
    & && f_i^P + g_i^P - \sum_{j\in \set{C}_i}(f_j^P + l_j r_j) = d_i^P & i\in\set{N} \label{eq:enerbal_with_loss_active}\\
    & && f_i^Q + g_i^Q - \sum_{j\in \set{C}_i}(f_j^Q + l_j x_j) = d_i^Q & i \in \set{N} \label{eq:enerbal_with_loss_reactive} \\
    & && u_i + 2(r_i f_i^P + x_i f_i^Q) + l_i(r_i^2 + x_i^2)  =  u_{\set{A}_i} \hspace{-2cm}& \nonumber \\ 
        & && & i\in\set{N}^+\label{eq:volt_with_losses} \\
    & && \frac{(f_i^P)^2 + (f_i^Q)^2}{u_i} \leq l_i & i\in\set{N}^+ \label{eq:current_squared} \\
    & && u_i^{\min} \leq u_i \leq u_i^{\max} & i\in\set{N}^+ \\
    & && (f_i^P)^2 + (f_i^Q)^2 \leq (S_i^{\max})^2 & i\in\set{N}^+ \\
    & && (f_i^P\!-\!l_i^P\!r_i)^2\!+\!(f_i^Q\!-\!l_i^Q\!x_i)^2\!\!\leq (S_i^{\max})^2 & \!i\in\set{N}^+.
\end{align}%
\label{mod:MOD_BF}%
\end{subequations}%
\allowdisplaybreaks[0]%
\noindent
Note that \cref{mod:MOD_BF} uses a  SOC-relaxed branch flow model, \cite{farivar2013branch,papavasiliou2018analysis}, which accounts for the power losses in  \cref{eq:enerbal_with_loss_active,eq:enerbal_with_loss_reactive}, and is modified to include decision variables $\alpha_i, i \in \set{G}$ and the (linear) deterministic equivalents of the generation chance constraints \cref{det-cc-opf:delta_i+,det-cc-opf:delta_i-} to compute reserve $z_{\epsilon_g}s\alpha_i,  i \in\set{G}$. The solution of \cref{mod:MOD_BF} is used  below as a linearization point to compute loss factors. We denote this linearization point as  $\{\bar{g}_i^{P},  i \in \set{G}; \bar{g}_i^{Q},  i \in \set{G}; \bar{f}_i^{P},  i \in \set{N}^+; \bar{f}_i^{Q},  i \in \set{N}^+; \bar{u}_i,  i \in \set{N}; \bar{l}_i,  i \in \set{N}^+\}$. 

Since the power losses of each edge $j$ in \cref{eq:enerbal_with_loss_active,eq:enerbal_with_loss_reactive}  are allocated to its upstream node $\set{A}_j$, terms  $\sum_{j\in\set{C}_i} l_j r_j$, $\sum_{j\in\set{C}_i} l_j x_j$ at each node $i$ can be interpreted as a (additional) fictitious nodal demand (FND) at node $i$, \cite{bai2018distribution}. 
To approximate the FND around the linearization point, we first obtain the sensitivity of current $\bar{l}_i(\bar{f}^P_i, \bar{f}^Q_i, \bar{u}_i)$ at the linearization point with respect to changes in active and reactive nodal net demand and production. 
Assuming that \cref{eq:current_squared} is tight at the optimum of \cref{mod:MOD_BF}, \cite{farivar2013branch}, we compute:
\begin{align}
    L_{ik}^P \coloneqq \frac{\partial \bar{l}_i}{\partial d_k^P} = - \frac{\partial \bar{l}_i}{\partial g_k^P} = \left(2 \bar{f}_i^P A_{ik} + 2 \bar{f}_i^Q A_{ik}  \right) \frac{1}{\bar{u}_i} \label{eq:current_sensit_active}\\
    L_{ik}^Q \coloneqq\frac{\partial \bar{l}_i}{\partial d_k^Q} = - \frac{\partial \bar{l}_i}{\partial g_k^Q} = \left(2 \bar{f}_i^P A_{ik} + 2 \bar{f}_i^Q A_{ik}  \right) \frac{1}{\bar{u}_i}, \label{eq:current_sensit_reactive}
\end{align}
where $L_{ik}^P$ and $L_{ik}^Q$ define the sensitivity of power losses of edge $i$ to  active and reactive power changes  at node~$k$.
Next, using \cref{eq:current_sensit_active,eq:current_sensit_reactive}, we can find the sensitivity of the active FND at node $i$ to active and reactive net demand deviations from the linearization point at  node $j$ as:
\begin{align}
    \LP_{ij}^P = \sum_{k\in\set{C}_i} L_{kj}^{P} r_k, \quad \LP_{ij}^Q = \sum_{k\in\set{C}_i} L_{kj}^{Q} r_k \label{eq:LP_entries}
\end{align}
and the sensitivity of reactive FND at node $i$ to active and reactive net demand deviations from the linearization point at  node $j$ as:
\begin{align}
    \LQ_{ij}^P = \sum_{k\in\set{C}_i} L_{kj}^{P} x_k, \quad \LQ_{ij}^Q = \sum_{k\in\set{C}_i} L_{kj}^{Q} x_k. \label{eq:LQ_entries}
\end{align}
Since the forecast demand is fixed,  the linearized FND only depends on the deviation of active and reactive production levels $(g^{P}_i - \bar{g}_i^{P}),  i \in \set{G}$, and $(g^{Q}_i - \bar{g}_i^{Q}),  i \in \set{G}$.
Thus, the  loss-aware nodal power balance constraints are given as:
\begin{align}
    &(\lambda^P_i): \!\!&&\!\! f_i^P + g_i^P -\!\sum_{j\in \set{C}_i}(f_j^P + \bar{l}_j r_j) + \Ploss_i(g^P, g^Q) = d_i^P  \label{eq:enerbal_with_approx_loss_active}\\
    &(\lambda^Q_i): \!\!&&\!\! f_i^Q - g_i^Q -\!\sum_{j\in \set{C}_i}(f_j^Q + \bar{l}_j x_j) + \Qloss_i(g^P, g^Q) = d_i^Q, \label{eq:enerbal_with_approx_loss_reactive}
\end{align}
where:
\begin{align}
    &\Ploss_i(g^P, g^Q) \coloneqq \sum_{j\in \set{G}}(\LP^P_{ij}(g_j^P - \bar{g}_j^P) + \LP^Q_{ij}(g_j^Q - \bar{g}_j^Q)) \\
    &\Qloss_i(g^P, g^Q) \coloneqq \sum_{j\in \set{G}}(\LQ^P_{ij}(g_j^P - \bar{g}_j^P) + \LQ^Q_{ij}(g_j^Q - \bar{g}_j^Q)).
\end{align}

To determine the impact of power losses with respect to uncertainty $\om$, we define matrices $\LP^P$ and $\LQ^P$ with elements $\LP^P_{ij}$ and $\LQ^P_{ij}$ given by \cref{eq:LP_entries,eq:LQ_entries}. 
Using these matrices, we define the loss-aware extensions of matrices $A$ and $R$ denoted as $A^L$ and $R^L$:
\begin{align}
    A^L &\coloneqq A (I + \LP^P) \label{eq:A_lossaware}\\
    R^L &\coloneqq A^{\!\top} (\diag(r) A^L + \diag(x) A \LQ^P) \label{eq:R_lossaware}.
\end{align}
Therefore, the loss-aware modification of the EQV-CC is obtained by substituting $A$ with $A^L$ and $R$ with $R^L$ and extending the nodal power balances as given in \cref{eq:enerbal_with_approx_loss_active,eq:enerbal_with_approx_loss_reactive}.
The AC-CCOPF model with power losses is presented in detail in Section~\ref{ssec:DLMP_with_losses}, see \eqref{mod:LVOLT-CC}.
}

\section{DLMPs with Chance-Constrained  Limits}
\subsection{DLMPs with Chance-Constrained  Generation Limits}
\label{ssec:DLMPS_with_chance_constrained_generation_limits}

In this subsection, we consider a modification of the  EQV-CC in \eqref{eq:eqvCC} that models chance constraints on the generation outputs in \eqref{det-cc-opf:delta_i+}--\eqref{det-cc-opf:delta_i-} and other constraints are considered deterministically. This modification is given below:
\allowdisplaybreaks
\begin{subequations} \label{eq:GEN-CC}
\begin{align}
    &\text{GEN-CC:} && \min_{\substack{\{g_i^P, g_i^Q,\alpha_i\}_{i\in N}, \\ \{f_i^P, f_i^Q, u_i\}_{i=N^+}}} 
    \sum_{i=0}^n\left(c_i(g_i^P) + \frac{\alpha_i^2}{2 b_i} s^2\right) 
\end{align}

\begin{align}
   \text{s.t.} & && \text{\cref{cc-opf:lambda_0,cc-opf:lambda_i,cc-opf:pi_0,cc-opf:pi_i,cc-opf:beta_i,cc-opf:theta_i+,cc-opf:theta_i-,cc-opf:gamma}} \nonumber \\
    & && \text{\cref{det-cc-opf:delta_i+,det-cc-opf:delta_i-}}  \nonumber \\
    &(\mu_i^+): && u_i \leq u_i^{\max} & i\in\set{N}^+ \label{GEN-CC:mu_i+}\\
    &(\mu_i^-): && -u_i \leq -u_i^{\min} & i\in\set{N}^+  \label{GEN-CC:mu_i-} \\
    &(\eta_i): && (f_i^P)^2 + (f_i^Q)^2 \leq (S_i^{\max})^2 & i\in\set{N}^+  \label{GEN-CC:eta_i}.
\end{align}
\label{mod:GEN-CC}
\end{subequations}
\allowdisplaybreaks[0]

We use the GEN-CC to compute the power and balancing regulation prices, which are given by dual multiplier $\lambda^P_i$ and $\lambda^Q_i$ of the power balance constraint \cref{cc-opf:lambda_i} and \cref{cc-opf:pi_i}, as well as dual multiplier $\gamma$ of the system-wide balancing regulation condition in  \cref{cc-opf:gamma}. Thus, we formulate and prove:

\begin{prop}
\label{prop:price_decomp}
Consider the GEN-CC in \eqref{eq:GEN-CC}. Let $\lambda^P_i$ and $\lambda^Q_i$ be the active and reactive power prices defined as dual multipliers of constraints \eqref{cc-opf:lambda_i} and \eqref{cc-opf:pi_i}. Then  $\lambda^P_i$ and $\lambda^Q_i$ are given by the following functions:
\allowdisplaybreaks
\begin{align}
    \lambda^P_i &= \lambda^P_{\set{A}_i} + (\lambda^Q_i - \lambda^Q_{\set{A}_i}) \frac{r_i}{x_i} 
    - 2\eta_i\left(f_i^P + \frac{r_i}{x_i} f_i^Q\right) \label{eq:active_power_price} \\
    \lambda^Q_i &= \lambda^Q_{\set{A}_i} + (\lambda^P_i - \lambda^P_{\set{A}_i}) \frac{x_i}{r_i} 
    - 2\eta_i\left(f_i^Q + \frac{x_i}{r_i} f_i^P\right) \label{eq:reactive_power_price_rec},
\end{align}
where $\eta_i$  is a dual multiplier of \eqref{GEN-CC:eta_i}.
\end{prop}
\allowdisplaybreaks[0] 
\begin{proof}
The Karush-Kuhn-Tucker (KKT) optimality conditions for the GEN-CC in \eqref{eq:GEN-CC} are:
\allowdisplaybreaks
\begin{subequations}
\begin{align}
    &(g_i^P): && \frac{(g_i + a_i)}{b_i} + \delta_i^+ - \delta_i^- - \lambda^P_i = 0 && i\in\set{G} \label{equiv-kkt:gP_i}\\
    &(g_i^Q): &&  \theta_i^+ - \theta_i^- - \lambda^Q_i = 0 && i\in\set{G} \label{equiv-kkt:gQ_i}\\
    &(u_i): && \beta_i - \sum_{j\in\set{C}_i} \beta_j + \mu_i^+ - \mu_i^- = 0 && i\in\set{N}^+ \label{equiv-kkt:u_i}\\
    &(f_i^P): && \lambda^P_i - \lambda^P_{\set{A}_i} + 2r_i \beta_i + 2 f_i^P \eta_i = 0 && i\in\set{N}^+ \label{equiv-kkt:fP_i}\\
    &(f_i^Q): && \lambda^Q_i - \lambda^Q_{\set{A}_i} + 2x_i \beta_i +  2 f_i^Q \eta_i= 0 && i\in\set{N}^+ \label{equiv-kkt:fQ_i}\\
    &(\alpha_i): && \frac{\alpha_i}{b_i} s^2 + z_{\epsilon_g} s(\delta_i^+ +\delta_i^-) - \gamma = 0 && i\in\set{N}^+ \label{equiv-kkt:alpha_i}\\
    &(\alpha_0): && \frac{\alpha_0}{b_0} s^2 - \gamma = 0 & \label{equiv-kkt:alpha_0}
\end{align}
\vspace*{-0.4cm}
\begin{align}
& 0 \leq \delta_i^+ \bot g_i^P + z_{\epsilon_g} \alpha_i s - g_i^{P,\max} \geq 0 && i\in\set{G} \label{equiv-kkt:delta_i+} \\ 
& 0 \leq \delta_i^- \bot -g_i^P + z_{\epsilon_g} \alpha_i s + g_i^{P,\min} \geq 0 && i\in\set{G} \label{equiv-kkt:delta_i-} \\ 
& 0 \leq \theta_i^+ \bot g_i^Q - g_i^{Q,\max} \geq 0 && i\in\set{G} \label{equiv-kkt:theta_i+}\\
& 0 \leq \theta_i^- \bot {-g_i^Q} + g_i^{Q,\min} \geq 0 && i\in\set{G} \label{equiv-kkt:theta_i-}\\
& 0 \leq \mu_i^+ \bot +u_i - u_i^{\max} \geq 0 && i\in\set{N}^+ \label{equiv-kkt:mu_i+}\\
& 0 \leq \mu_i^- \bot {-u_i} + u_i^{\min} \geq 0 && i\in\set{N}^+ \label{equiv-kkt:mu_i-} 
\\ & 0 \leq \eta_i \bot (f_i^P)^2 + (f_i^Q)^2 - (S_i^{\max})^2 \geq 0 && i\in\set{N}^+ \label{equiv-kkt:eta_i}
\end{align}%
\end{subequations}%
\allowdisplaybreaks[0]%
Expressing $\lambda^P_i$ and $\lambda^Q_i$ from \cref{equiv-kkt:fP_i,equiv-kkt:fQ_i} instantly yields the expressions in  \cref{eq:active_power_price,eq:reactive_power_price_rec}.
\end{proof}

\begin{remark}
\textcolor{black}{ 
Eqs.~\cref{eq:active_power_price,eq:reactive_power_price_rec} can also be used to couple DLMPs and  transmission locational marginal prices (LMP) for active and, if available,  reactive power obtained from wholesale market-clearing outcomes. Indeed, transmission LMPs can be  parameterized in \cref{eq:active_power_price,eq:reactive_power_price_rec} as prices at the root node, i.e. $\lambda^P_{0}$ and $\lambda^Q_{0}$.}
\end{remark}

Proposition \ref{prop:price_decomp} allows for multiple  insights on the price formation process. First, both $\lambda^P_i$  and $\lambda^Q_i$ do not explicitly depend on uncertainty and risk parameters. Next, as the second terms in \eqref{eq:active_power_price}  and \eqref{eq:reactive_power_price_rec}  reveal, $\lambda^P_i$  and $\lambda^Q_i$ are mutually dependent. Furthermore,  the third terms in  \eqref{eq:active_power_price}  and \eqref{eq:reactive_power_price_rec} demonstrates that $\lambda^P_i$  and $\lambda^Q_i$ are both equally dependent on active and reactive power flows  $f^P_i$  and $f^Q_i$, as well as edge characteristics $r_i$ and $x_i$. Finally, if the distribution system is not power-flow-constrained, i.e. $\eta_i=0$ and \eqref{GEN-CC:eta_i} is not binding, the third terms disappear  in  \eqref{eq:active_power_price}  and \eqref{eq:reactive_power_price_rec}. However, even in this case the DLMPs at different nodes would not be the same due the need to provide both reactive and active power.  

{\color{black} 
Since Proposition~\ref{prop:price_decomp} relates prices $\lambda_i^P$ and $\lambda_i^Q$ at neighboring nodes, it implies that changing a real power injection at any node $i$ can be compensated by active and reactive power  adjustments at either the ancestor node or that node without any other changes in the system.
However,  similarly to the discussion in  \cite{papavasiliou2018analysis}, the physical dependencies between the nodes are more complex and net injection changes at one node will shift operating conditions at all other nodes. 
This becomes clear when we reinterpret the results of Proposition~\ref{prop:price_decomp} in terms of the voltage limits given by \cref{GEN-CC:mu_i+,GEN-CC:mu_i-}.
} 
For this purpose we express $\beta_i$ from  \eqref{equiv-kkt:u_i} and use it in \eqref{equiv-kkt:fP_i} and \eqref{equiv-kkt:fQ_i}. Expressing  $\lambda^P_i$ and  $\lambda^Q_i$ from \eqref{equiv-kkt:fP_i} and \eqref{equiv-kkt:fQ_i} leads to:
\begin{align}
    \lambda^P_i = \lambda^P_{\set{A}_i} - 2r_i\sum_{j\in\set{D}_i}(\mu_j^+ - \mu_j^-) + 2 f_i^P \eta_i \label{eq:lambda_i_decomp}\\
    \lambda^Q_i = \lambda^Q_{\set{A}_i} - 2x_i\sum_{j\in\set{D}_i}(\mu_j^+ - \mu_j^-) + 2 f_i^Q \eta_i. \label{eq:pi_i_decomp}
\end{align}
Thus, if voltage limits are binding at downstream nodes $j\in \mathcal{D}_i$ of node $i$, i.e. $\mu_j^+\neq0$ or $\mu_j^-\neq 0$, they will contribute to the resulting values of   $\lambda^P_i$ and  $\lambda^Q_i$. Furthermore, expressions in \eqref{eq:lambda_i_decomp} and \eqref{eq:pi_i_decomp} show that if the distribution system is not voltage- or power-flow-congested, i.e. ${\mu_i^+ = \mu_i^- =\eta_i=0,  i \in \mathcal N}$, DLMPs reduce to system-wide prices equal to the prices at the root node, i.e.  $\lambda^P_i=  \lambda^P_0$ and  $\lambda^Q_i=  \lambda^Q_0$.

Unlike $\lambda^P_i$ and  $\lambda^Q_i$, we find that  the price for balancing regulation explicitly depends on uncertainty and risk parameters:
\begin{prop}
\label{prop:gamma}
Consider the GEN-CC  in \eqref{eq:GEN-CC}. Let $\gamma$ be the  balancing regulation price defined as a dual multiplier of constraint \eqref{cc-opf:gamma}. Then  the following function defines $\gamma$:
\begin{align}
    \gamma = \frac{s}{\sum_{i=0}^n b_i}\left( s + z_\epsilon  \sum_{i=1}^n (\delta_i^+ + \delta_i^-)b_i \right). 
    \label{eq:gamma}
\end{align}

\end{prop}
\begin{proof}
Expressing $\alpha_i$ and $\alpha_0$ from \eqref{equiv-kkt:alpha_i}  and \eqref{equiv-kkt:alpha_0} in terms of $\gamma$ and using it \eqref{cc-opf:gamma} yields:
\begin{align}
    1 + \frac{\gamma b_0}{s^2} = -\sum_{i=1}^n\left[\gamma + z_{\epsilon_g} s(\delta_i^+ + \delta_i^-)\right] \frac{b_i}{s^2}, 
\end{align}
which immediately leads to  \cref{eq:gamma}.
\end{proof}

As per  \eqref{eq:gamma}, $\gamma$ depends on uncertainty, since ${s^2 \coloneqq e^{\!\top}\Sigma e}$, as well as  risk tolerance of the DSO, since $z_{\epsilon_g} = \Phi^{-1} (1-\epsilon_g)$. 
Notably, the balancing regulation price is always non-zero if there is uncertainty in the system (i.e. $s \neq 0$).
This is true even if none of the chance constraints on output limits of DERs in \eqref{cc-opf:delta_i+}-\eqref{cc-opf:delta_i-} are binding, i.e.  $\delta_i^{+}=\delta_i^{-}=0,  i \in \mathcal N$. 
In other words, as long as the forecast is not perfect, there is a value on procuring a non-zero amount of balancing regulation.


\begin{figure}[b]
    \centering
    \includegraphics[width=0.95\linewidth]{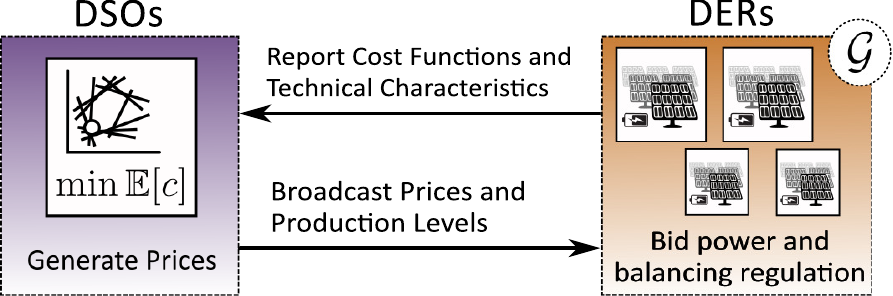}
    \caption{A schematic representation of the auction.}
    \label{fig:auction_illusttration}
\end{figure}

{\color{black} 
The prices resulting from Propositions~\ref{prop:price_decomp} and~\ref{prop:gamma} can be leveraged by the DSO to organize a stochastic distribution electricity market, e.g. via a centralized auction.
Fig.~\ref{fig:auction_illusttration} illustrates such an auction, where, first, producers  truthfully report their cost functions and technical characteristics to the DSO. 
Next, the DSO determines the optimal dispatch decisions and resulting prices for each market product \textcolor{black}{using the best available forecast information}. 
Finally, these decisions and prices are communicated by the DSO to all producers.
For this stochastic electricity market, we define a \textit{competitive equilibrium} as a set of production levels and prices $\{g_i^{P}, i \in \set{G}; \alpha_i,i \in \set{G}; \pi^g_i, i \in \set{G}; \pi^\alpha\}$ that (i) clears the market so that the production and demand quantities are balanced and $\sum_{i\in\set{G}}\alpha_i =1$ and (ii) maximizes the profit of all producers, under the market payment structured as $\pi^g g^p_i  + \pi^\alpha \alpha_i$,  so that there is no incentive to deviate from the market outcomes. 
}

To show that the prices from Propositions~\ref{prop:price_decomp} and~\ref{prop:gamma} support the competitive equilibrium, we consider the GEN-CC in \eqref{eq:GEN-CC} and the  behavior of each producer (controllable DER) is modeled as \textcolor{black}{a risk-neutral}, profit-maximization: 
\begin{align}
& \Big\{ \max_{g_i^P, \alpha_i}  \Pi_i = \overbrace{\pi^g_i g_i^P  + \pi^\alpha \alpha_i\vphantom{\frac{s^2}{2 b_i}}}^{\text{\textcolor{black}{Payment}}}  \overbrace{- c_i(g_i^P)   -\alpha_i^2 \frac{s^2}{2 b_i}}^{\text{Cost}}   \label{eq:individual_generator} \\
&\text{s.t} ~ (\delta_i^-,\delta_i^+): ~
g_i^{P,\min}\! +\! z_{\epsilon_g} \alpha_i s \leq g_i^P \leq g_i^{P,\max}\!-  z_{\epsilon_g} \alpha_i s  \Big\} \nonumber  \\
& \hspace{7.5cm}  i \in \set{G}, \nonumber 
\end{align}
where $\Pi_i$  denotes the profit function of each controllable DER at node $i$ and $\{\pi^g_i, \pi^\alpha\}$  are active power and balancing regulation prices.  

{\color{black}
\begin{remark}
Since uncertainty and risk parameters, i.e. $s^2=e^{\!\top}\Sigma e$ and $z_{\epsilon_g}$, are shared by the DSO and producers, we assume that this knowledge is common and consensual. Although these parameters can be exploited by  the DSO to advance their self-interest and increase security margins above reasonable levels at the expense of customers, this behavior can be mitigated using benchmarking and performance-based rate design practices, \cite{wood2016recovery,utilityofthefuture,farsi2005benchmarking}. 
\end{remark}
}

Considering this stochastic market, as in Fig.~\ref{fig:auction_illusttration}, we prove:

\begin{theorem}
\label{prop:competive_equilibrium}
Let $\{g_i^{P,*}, \alpha_i^*,  i \in \mathcal G\}$ be an optimal solution of  the GEN-CC  in \eqref{eq:GEN-CC} and let $\{\lambda^{P,*}_i, i \in \mathcal N; \gamma^*\}$ be the  dual variables of \eqref{cc-opf:lambda_i} and \eqref{cc-opf:pi_i}, 
{\color{black} 
then the set of production levels and prices $\{g_i^{P,*}, i \in \set{G}; \alpha_i^*,i \in \set{G}; \pi^g_i, i \in \set{G}; \pi^\alpha\}$ is a competitive equilibrium if $\pi_i^g = \lambda^{P,*}_i, i \in \mathcal G,$ and $\pi^{\alpha} = \gamma^*$.}
\end{theorem}
\vspace{-0.4cm}
\begin{proof}
The KKT optimality conditions for \eqref{eq:individual_generator} are:
\begin{subequations}
\begin{align}
    &(g_i^P): && \frac{(g_i + a_i)}{b_i} + \delta_i^+ - \delta_i^- - \pi_i^g = 0  \label{eq:kkt_g_i_ind} \\
    &(\alpha_i): && \frac{\alpha_i}{b_i} s^2 + z_{\epsilon_g} s(\delta_i^+ +\delta_i^-) - \pi^{\alpha} = 0 \label{eq:kkt_alpha_i_ind}\\
    & &&0 \leq \delta_i^+ \bot g_i^P + z_{\epsilon_g} \alpha_i s - g_i^{P,\max} \geq 0  \label{eq:kkt_delta_plus_ind} \\ 
    & && 0 \leq \delta_i^- \bot -g_i^P + z_{\epsilon_g} \alpha_i s + g_i^{P,\min} \geq 0. \label{eq:kkt_delta_minus_ind}
\end{align}%
\label{eq:ind_kkts}%
\end{subequations}%
Using \eqref{eq:kkt_g_i_ind}  and \eqref{eq:kkt_alpha_i_ind}, we express $\pi^{g}_i = - \frac{(g_i - a_i)}{b_i} - \delta_i^+ +\delta_i^-$  and $\pi^{\alpha} = -\frac{\alpha_i}{b_i} s^2 - z_{\epsilon_g} s(\delta_i^+ +\delta_i^-)$. Similarly, we express $\lambda^P_i$ and $\gamma$ from  \cref{equiv-kkt:gP_i,equiv-kkt:alpha_i}. Therefore, $\lambda^P_i = - \frac{(g_i - a_i)}{b_i} - \delta_i^+ +\delta_i^-= \pi^{g}_i$  and $\gamma = -\frac{\alpha_i}{b_i} s^2 - z_{\epsilon_g} s(\delta_i^+ +\delta_i^-)= \pi^{\alpha}$.  If $\{g_i^{P,*}, \alpha_i^*,  i \in \set{G}\}$, it follows that    $\lambda^{P,*}_i = \pi^{g}_i$ and $\gamma^* = \pi^{\alpha}_i$, i.e. $\{g_i^{P,*}, i \in \set{G}; \alpha_i^*,i \in \set{G}; \pi^g_i, i \in \set{G}; \pi^\alpha\}$ solves \eqref{eq:individual_generator} and maximizes $\Pi_i$. Therefore, $\{g_i^{P,*}, i \in \set{G}; \alpha_i^*,i \in \set{G}; \pi^g_i, i \in \set{G}; \pi^\alpha\}$ is a competitive equilibrium.
\end{proof}

\textcolor{black}{Since both the DSO and producers are modeled as risk-neutral, see \eqref{eq:ccacopf} and \eqref{eq:individual_generator}, and share common knowledge about underlying uncertainty parameters,  the competitive equilibrium  established by Theorem~\ref{prop:competive_equilibrium}  also corresponds to the welfare-maximization (cost-minimization) solution, \cite{doi:10.1137/110851778}. Notably, this property will hold as long as the DSO and producers continue sharing common knowledge about underlying uncertainty parameters,  even if their attitudes toward risk vary based on a given coherent risk measure, \cite{doi:10.1137/110851778}.  Although, from the viewpoint of customers, internalizing the uncertainty and risk parameter in the equilibrium prices from Theorem~\ref{prop:competive_equilibrium} may increase electricity prices relative to the deterministic case, stochasticity-aware prices will  provide incentives to reduce their uncertainty, thus reducing balancing regulation needs in the system,  or  to  exercise  more  flexibility  (e.g.  to  shift  their demand to  time periods with lower DLMPs).}

\textcolor{black}{Hence, using the competitive equilibrium of  Theorem~\ref{prop:competive_equilibrium},} we can analyze the effect of the  prices on the capacity allocation  between the power production and balancing regulation from the perspective of each  producer modeled as in \eqref{eq:individual_generator}. Let $\{\pi^g_i, \pi^\alpha\}$ be given prices and let $\{g_i^{P,*}, \alpha_i^*\}$ be the optimal solution of \eqref{eq:individual_generator} for these prices.  The KKT optimality conditions in  \cref{eq:kkt_g_i_ind,eq:kkt_alpha_i_ind,eq:kkt_delta_plus_ind,eq:kkt_delta_minus_ind} can be used to find parametric functions that determine the optimal  dispatch of each controllable DER. These functions depend on whether constraints in \eqref{eq:individual_generator} are binding or not. Since \eqref{eq:individual_generator} has two inequality constraints, we consider the following four cases:
\begin{enumerate}
    \item  $\delta_i^{+,\ast} = \delta_i^{-,\ast} = 0$: When \eqref{eq:individual_generator} has no binding constraints, it follows from \cref{eq:kkt_g_i_ind,eq:kkt_alpha_i_ind} that:
    \begin{align}
      &    g_i^{P,\ast} = \pi_i^g b_i - a_i, \quad  \alpha^\ast_i = \pi^\alpha b_i/s^2
     \label{eq:uncontrainted}
    \end{align}
    Inserting the optimal dispatch given by \eqref{eq:uncontrainted} into  \cref{eq:kkt_delta_plus_ind,eq:kkt_delta_minus_ind} leads to the following relationship between prices $\pi_i^g$ and $\pi^\alpha$:
    \begin{align}
        \frac{g_i^{\min} + a_i}{b_i} + z_{\epsilon_g} \frac{\pi^\alpha }{s}  \leq \pi^g_i  \leq \frac{g^{\max} + a_i}{b_i} -  z_{\epsilon_g} \frac{\pi^\alpha}{s}. \label{eq:condition_firstcase} 
    \end{align}

    \item  $\delta_i^{+,\ast}\neq 0, \delta_i^{-,*} = 0$: Since $\delta_i^{+,\ast}\neq 0$, only the upper limit is binding. Thus, \eqref{eq:kkt_delta_plus_ind} yields $g_i^{P,\ast} + z_{\epsilon_g} \alpha_i^{\ast} s - g_i^{P,\max} =0$, which in combination with \cref{eq:kkt_g_i_ind,eq:kkt_alpha_i_ind} leads to:
    \begin{subequations} \label{eq:upbind}
    \begin{align}
      &    g_i^{P,\ast} = g_i^{\max} - z_{\epsilon_g} s \alpha_i^*  \\
      &    \alpha_i^* = \frac{z_{\epsilon_g} s (g^{\max} + a_i - b_i \pi_i^g) + \pi^\alpha b_i}{s^2(1+z_{\epsilon_g}^2)}. 
    \end{align}
    \end{subequations}
    With the upper constraint binding,  it follows from \eqref{eq:condition_firstcase} that \eqref{eq:upbind} holds if:
    \begin{align}
        \pi^g_i  \geq \frac{g^{\max} + a_i}{b_i} -  z_{\epsilon_g} \frac{\pi^\alpha}{s}.
    \end{align}

    \item $\delta_i^{+,\ast} = 0, \delta_i^{-,*} \neq0$: This case is the opposite of the previous one since only the lower limit is binding. Therefore, \eqref{eq:kkt_delta_minus_ind} yields $-g_i^{P,\ast} + z_{\epsilon_g} \alpha_i^{\ast} s + g_i^{P,\min} =0$, which in combination with \cref{eq:kkt_g_i_ind,eq:kkt_alpha_i_ind} leads to:
    \begin{subequations}
    \begin{align}
        & g_i^{P,\ast} = g_i^{\min} + z_{\epsilon_g} s \alpha_i^*  \\
        &   \alpha_i^* = \frac{z_{\epsilon_g} s (g^{\min} + a_i - b_i \pi_i^g) - \pi^\alpha b_i}{s^2(1+z_{\epsilon_g}^2)}. 
    \end{align}
    \label{eq:lowbind}
    \end{subequations}
    With the lower constraint binding it follows from \eqref{eq:condition_firstcase} that \eqref{eq:lowbind} holds if:
    \begin{align}
         \pi^g_i \leq \frac{g_i^{\min} + a_i}{b_i} + z_{\epsilon_g} \frac{\pi^\alpha }{s}. 
    \end{align}

    \item $\delta_i^{+,\ast}\neq 0, \delta_i^{-,*} \neq 0 $:
    When both constraints of \eqref{eq:individual_generator} are binding it follows from \cref{eq:kkt_delta_plus_ind,eq:kkt_delta_minus_ind} that:
    \begin{align}
        & g_i^{P,\ast} =  \frac{g_i^{\max} + g_i^{\min}}{2}, \quad \alpha_i^* = \frac{g_i^{\max} - g_i^{\min}}{2z_{\epsilon_g}s},
        \label{eq:bothbind}
    \end{align}
    where $g_i^*$  is the midpoint of the dispatch range and the upward $(g_i^{\max}-g_i^*)$  and downward ($g_i^*-g_i^{\min}$) margins are fully used for providing balancing regulation. In this case,  it follows from \eqref{eq:condition_firstcase} that $\pi^\alpha$ is independent of $\pi^g_i$ and must be as follows:
    \begin{align}
        \pi^\alpha \geq s(g_i^{\max} - g_i^{\min})/2 z_{\epsilon_g} b_i. \label{eq:bothbind_cond}
    \end{align}
    
\end{enumerate}
The  dispatch policies in \cref{eq:uncontrainted,eq:upbind,eq:lowbind,eq:bothbind} support the competitive  equilibrium established by Theorem~\ref{prop:competive_equilibrium} and can be implemented locally at each DER, if there is communication to broadcast prices $\pi_i^g$ and $\pi^\alpha$.

{\color{black}
\begin{remark}
Eqs.~\cref{eq:uncontrainted}--\cref{eq:bothbind_cond} provide a parametric model of the reaction of each producer to given price signals. 
By observing the provided production levels and balancing participation factors for given prices over time, the DSO can use machine learning methods to estimate these parameters.
This enables the DSO to either verify reported cost functions and technical characteristics or to establish a one-way communication market framework as implemented in \cite{mieth2018online} for demand-side management.
\end{remark}
}

\subsection{DLMPs with Chance-Constrained Voltage Limits}

The the GEN-CC in \eqref{eq:GEN-CC} has deterministic voltage limits as given by \eqref{GEN-CC:mu_i+} and \eqref{GEN-CC:mu_i-}. We recast these limits as chance constraints, which leads to the following optimization:
\allowdisplaybreaks
\begin{subequations}
\begin{align}
    &\text{VOLT-CC:} \min_{\substack{\{g_i^P, g_i^Q,\alpha_i\}_{i\in N}, \\ \{f_i^P, f_i^Q, u_i\}_{i=N^+}}} 
    \sum_{i=0}^n\left(c_i(g_i^P) + \frac{\alpha_i^2}{2 b_i} s^2\right) 
\end{align}
\begin{align}
    \text{s.t.} & && \text{\cref{cc-opf:lambda_0,cc-opf:lambda_i,cc-opf:pi_0,cc-opf:pi_i,cc-opf:beta_i,cc-opf:theta_i+,cc-opf:theta_i-,cc-opf:gamma}} \nonumber \\
    & && \text{\cref{det-cc-opf:delta_i+,det-cc-opf:delta_i-,det-cc-opf:zeta_i,det-cc-opf:nu_i,det-cc-opf:mu_i+,det-cc-opf:mu_i-,GEN-CC:eta_i}}. \nonumber
\end{align}
\label{mod:VOLT-CC}
\end{subequations}
\allowdisplaybreaks[0]
Similarly to the GEN-CC in \eqref{eq:GEN-CC} we formulate and prove for the VOLT-CC in \eqref{mod:VOLT-CC} the following proposition:
\begin{prop}
\label{prop:gamma_with_voltCC}
Consider the VOLT-CC in \cref{mod:VOLT-CC}. Let $\lambda^P_i$, $\lambda^Q_i$ and $\gamma$ be the active power, reactive power and balancing regulation prices at node $i$. Then $\lambda^P_i$, $\lambda^Q_i$ are given by \eqref{eq:lambda_i_decomp}, \eqref{eq:pi_i_decomp} and $\gamma$ is given by
\begin{equation}
\begin{split}
    \gamma = \frac{s}{\sum_{i=0}^n b_i} \left(
     s +z_\epsilon  \sum_{i=1}^n (\delta_i^+ + \delta_i^-)b_i 
      + \sum_{i=1}^n b_i \nu_i^v \right),
    \label{eq:gamma_with_volt_cc}
\end{split}
\end{equation}
where $\nu_i^v$ is the dual multiplier of \cref{det-cc-opf:nu_i} given as:
\begin{align}
    \nu_i^v = 
      2  z_{\epsilon_v}  \sum_{j=1}^n R_{ji}(\mu_j^+ + \mu_j^-)
     \frac{ R_{j*} (\Sigma e + s^2 \alpha) }{\Stdv[\bm{u}_j(\om,\alpha)]}.
     \label{eq:nu_decomp}
\end{align}
\end{prop}

\begin{proof}
The KKT optimality conditions for \eqref{mod:VOLT-CC} are:
\allowdisplaybreaks
\begin{subequations}
\begin{align}
    & && \text{\cref{equiv-kkt:gP_i,equiv-kkt:gQ_i,equiv-kkt:u_i,equiv-kkt:fP_i,equiv-kkt:fQ_i,equiv-kkt:alpha_0,equiv-kkt:delta_i+,equiv-kkt:delta_i-,equiv-kkt:theta_i+,equiv-kkt:theta_i-,equiv-kkt:eta_i}} \hspace{-4cm} && \nonumber
 \\
    &(\alpha_i): && \frac{\alpha_i}{b_i} s^2 + z_\epsilon s(\delta_i^+ +\delta_i^-) - \gamma + \nu_i^v = 0 \hspace{-4cm} && \nonumber \\
    & && &&  i\in\set{N}^+ \label{volt-kkt:alpha_i}\\
    &(t_i^v): && 2 z_\epsilon (\mu_i^+ + \mu_i^-) - \zeta_i = 0 && i\in\set{N}^+ \label{volt-kkt:t_i^v}\\
    &(\rho_i^v): && \sum_{j=1}^n \nu^v_j \check{R}_{ji}  + \zeta_i \frac{ (R_{i*} + \rho_i^v e^{\!\top})\Sigma e }{t^v_i} = 0 \hspace{-4cm} && \nonumber \\
    & && && i\in\set{N}^+ \label{volt-kkt:rho_i^v} \\
    & && 0 \leq \mu_i^+ \bot u_i + 2 z_{\epsilon_v} t_i^v - u_i^{\max} \geq 0 && i\in\set{N}^+ \label{volt-kkt:mu_i+}\\
    & && 0 \leq \mu_i^- \bot -u_i 2 z_{\epsilon_v} t_i^v + u_i^{\min} \geq 0. && i\in\set{N}^+ \label{volt-kkt:mu_i-}
\end{align}%
\end{subequations}%
\allowdisplaybreaks[0]%
\noindent
It follows that the expressions for $\lambda_i^P$, $\lambda_i^Q$ are equal to the results of Proposition~\ref{prop:price_decomp}.
Expression~\cref{eq:gamma_with_volt_cc} is obtained analogously to the proof of Proposition~\ref{prop:gamma}.
To find \cref{eq:nu_decomp} we first express $\rho_i$ from \cref{det-cc-opf:nu_i} and $\zeta_i$ from \cref{volt-kkt:t_i^v} and insert these expressions into \cref{volt-kkt:rho_i^v}.
Second, if $\zeta_i \neq 0$, then \cref{det-cc-opf:zeta_i} is tight which means $t^v_i =  \Stdv[\bm{u}_j(\om,\alpha)]$ as per \cref{eq:voltage_stdev_conic}.
Finally, given that $\check{R} = R^{-1}$ as shown in Section~\ref{sssec:voltage_chance_constraints}, \cref{volt-kkt:rho_i^v} can be recast as~\cref{eq:nu_decomp}.
\end{proof}
{\color{black} 
Proposition~\ref{prop:gamma_with_voltCC} highlights the difficulty of enforcing probabilistic guarantees on system constraints (e.g. voltage limits) through such individual price signals.
While the structure of prices $\lambda_i^P$, $\lambda_i^Q$ does not change relative to Proposition~\ref{prop:price_decomp}, price $\gamma$ in \eqref{eq:gamma_with_volt_cc} includes $\sum_{i=1}^n b_i \nu_i^v$ in addition to the terms in \eqref{eq:gamma}.
This additional term leads to a discrepancy between the amounts of balancing participation deemed optimal by the DSO, which seeks to minimize the system-wide operating cost,  and by individual producers, which seek to maximize their individual profit.}
Notably, the expression for $\nu_i^v$ in \eqref{eq:nu_decomp} depends on  vector $\alpha$, which includes participation factors at all nodes.  Hence, introducing voltage chance constraints makes balancing regulation price $\gamma$ dependent on the choice of participation factors at all nodes and cannot be explained by purely local or neighboring voltage conditions, even in radial networks. Thus, if node $i$ is such that it has a high influence on the voltage magnitudes at other nodes (i.e. as captured by matrix $R$, see Eq.~\eqref{eq:voltage_reformulated}), the controllable DER at this node is implicitly discouraged from providing balancing regulation and, therefore,  $\nu_i^v$ drives the optimal choice of $\alpha_i$ from the system perspective.  However, since $\nu_i^v$  is not part of  \eqref{eq:individual_generator} and thus uncontrolled by DERs, it will not affect \cref{eq:ind_kkts}. This result shows that internalizing the effect of stochasticity on voltage limits,  which are enforced by the DSO and by producers, will prevent the existence of a competitive equilibrium enforced by Theorem~\ref{prop:competive_equilibrium} and, in this case, balancing participation price $\gamma$ must be adjusted to reflect this difference between the decision-making process of the DSO and controllable DERs.  
Assume $\alpha_i^{*,DSO}$ is the optimal amount of balancing regulation determined by the DSO by solving VOLT-CC. 
If the DSO brodcasts $\pi^\alpha = \gamma^*$ then DER $i$ will decide on its optimal participation $\alpha_i^{*,DER}$ by solving \cref{eq:individual_generator}.
The resulting difference between those balancing participation factors can then be quantified as:
\begin{align}
    \alpha_i^{*,DER} - \alpha_i^{*,DSO} = \frac{b_i}{s} \nu_i^v  \label{difference}
\end{align}
Note that \eqref{difference} is inversely proportional to the total uncertainty in the distribution system (note that $s= \sqrt{e\Sigma e^{\!\top}}$), i.e. the discrepancy between the DER and DSO perspectives decreases as more uncertainty is observed.

{\color{black} 
\subsection{DLMPs with losses}
\label{ssec:DLMP_with_losses}
To asses the effect of power losses on DLMPs, we consider the following optimization problem: 
\allowdisplaybreaks
\begin{subequations}
\begin{align}
    &\text{LVOLT-CC:} \min_{\substack{\{g_i^P, g_i^Q,\alpha_i\}_{i\in N}, \\ \{f_i^P, f_i^Q, u_i\}_{i=N^+}}} 
    \sum_{i=0}^n\left(c_i(g_i^P) + \frac{\alpha_i^2}{2 b_i} s^2\right) 
\end{align}
\begin{align}
    \text{s.t.} & && \text{\cref{eq:enerbal_with_approx_loss_active,eq:enerbal_with_approx_loss_reactive,cc-opf:beta_i,cc-opf:theta_i+,cc-opf:theta_i-,cc-opf:gamma}}\hspace{-4cm}  && \nonumber \\
    & && \text{\cref{det-cc-opf:delta_i+,det-cc-opf:delta_i-,det-cc-opf:mu_i+,det-cc-opf:mu_i-,GEN-CC:eta_i}} \hspace{-4cm} && \nonumber\\
     &(\zeta_i): && t_i^v \geq \norm{(R_{i*}^L + \rho_i^v e^{\!\top})\Sigma^{1/2}}_2 && i\in\set{N}^+ \label{lvolt-cc-opf:zeta_i} \\
    &(\nu_i^v): && \sum_{j=1}^n \check{R}_{ij}^L \rho_j^v = \alpha_i, && i\in\set{G} \label{lvolt-cc-opf:nu_i}
\end{align}
\label{mod:LVOLT-CC}
\end{subequations}
\allowdisplaybreaks[0]%
\noindent
where $\check{R}^L \coloneqq (R^L)^{-1}$ and claim:
\begin{prop}
\label{prop:including_losses}
Consider the LVOLT-CC in \cref{mod:LVOLT-CC}. Let $\lambda_i^P$, $\lambda_i^Q$ and $\gamma$ be the active power, reactive power and balancing regulation prices at node $i$. Then:
\begin{enumerate}[a)]
    \item Prices $\lambda_i^P$, $\lambda_i^Q$ are given by \eqref{eq:lambda_i_decomp}, \eqref{eq:pi_i_decomp} and $\gamma$ is given by:
    \begin{equation}
    \begin{split}
        \gamma = \frac{1}{\sum_{i=0}^n b_i} \left(
         s^2 +z_\epsilon s \sum_{i=1}^n (\delta_i^+ + \delta_i^-)b_i 
          + \sum_{i=1}^n b_i \nu_i^v \right),
        \label{eq:gamma_with_lvolt_cc}
    \end{split}
    \end{equation}
    where:
    \begin{align}
        \nu_i^v =  2  z_{\epsilon_v}   \sum_{j=1}^n R^L_{ji}(\mu_j^+ + \mu_j^-)
         \frac{ R^L_{j*} (\Sigma e + s^2 \alpha) }{\Stdv[\bm{u}_j(\om,\alpha)]}.
         \label{eq:nu_decomp_losses}
    \end{align}
    
    \item The optimal active production level $g_i^{P,*}$ is:
    \begin{align}
       & g_i^{P,*} = b_i(\lambda_i^P - (\delta_i^+ - \delta_i^-) + \xi_i^P(\lambda^P, \lambda^Q)) - a_i,
    \end{align}
    where
    \begin{align}
        & \xi_i^P(\lambda^P, \lambda^Q) \coloneqq \sum_{j=1}^N \LP^P_{ji} \lambda_j^P + \sum_{j=1}^N \LQ^P_{ji} \lambda_j^Q  \label{eq:xi_p}
    \end{align}
\end{enumerate}
\end{prop}

\begin{proof}
Consider the KKT optimality conditions for \cref{mod:LVOLT-CC}:
\allowdisplaybreaks
\begin{subequations}
\begin{align}
    & && \text{\cref{equiv-kkt:u_i,equiv-kkt:fP_i,equiv-kkt:fQ_i,equiv-kkt:alpha_0,equiv-kkt:delta_i+,equiv-kkt:delta_i-,equiv-kkt:theta_i+,equiv-kkt:theta_i-,equiv-kkt:eta_i,volt-kkt:mu_i+,volt-kkt:mu_i-}} \hspace{-4cm} &&
\nonumber \\
    &(g_i^P): && 
        \frac{(g_i + a_i)}{b_i} + \delta_i^+ - \delta_i^- - \lambda^P_i  && \label{l-volt-kkt:gi_P} \\ 
        & && \qquad - \underbrace{\sum_{j=1}^N \LP^P_{ji} \lambda_j^P + \sum_{j=1}^N \LQ^P_{ji} \lambda_j^Q}_{\xi_i^P(\lambda^P, \lambda^Q)} = 0 && i\in\set{G} \nonumber \\
    &(g_i^Q): && 
        \theta_i^+ - \theta_i^- - \lambda^Q_i \label{l-volt-kkt:gi_Q}\\
        & && - \sum_{j=1}^N \LP^P_{ji} \lambda_j^P + \sum_{j=1}^N \LQ^P_{ji} \lambda_j^Q = 0 && i\in\set{G}
        \nonumber \\
    &(\rho_i^v): && 
        \sum_{j=1}^n \eta^v_j \check{R}_{ji}^L  + \zeta_i \frac{ (R_{i*}^L + \rho_i^v e^{\!\top})\Sigma e }{t_i} = 0. \hspace{-4cm} &&  \label{l-volt-kkt:rho_i_v} \\
    & && && i\in\set{N}^+  \nonumber
\end{align}%
\end{subequations}%
\allowdisplaybreaks[0]%
Our result in Proposition 4a) follows directly from the proofs of Propositions~\ref{prop:price_decomp} and~\ref{prop:gamma_with_voltCC} with $R$ replaced by $R^L$ (see \cref{l-volt-kkt:rho_i_v}). Then, our result in Proposition 4b) follows directly from \cref{l-volt-kkt:gi_P,l-volt-kkt:gi_Q}.
\end{proof}

The term in \cref{eq:xi_p} relates the DLMP and optimal production level at node $i$ to the DLMPs at all other nodes via the loss factors.
For example, if production at $i$ has a high impact on active power losses at node $j$ (given by $\LP^P_{ji}$) and DLMP $\lambda_j^P$ is high, then active power production at node $i$ is discouraged by a lower DLMP $\lambda_i^P$.
Similarly to Proposition~\ref{prop:gamma_with_voltCC}, Proposition~\ref{prop:including_losses} reveals that power losses distort a competitive equilibrium because they are not part of the individual producers decisions.
}

\section{Case Study}

\begin{figure}[b]
\centering
\includegraphics[width=1\linewidth]{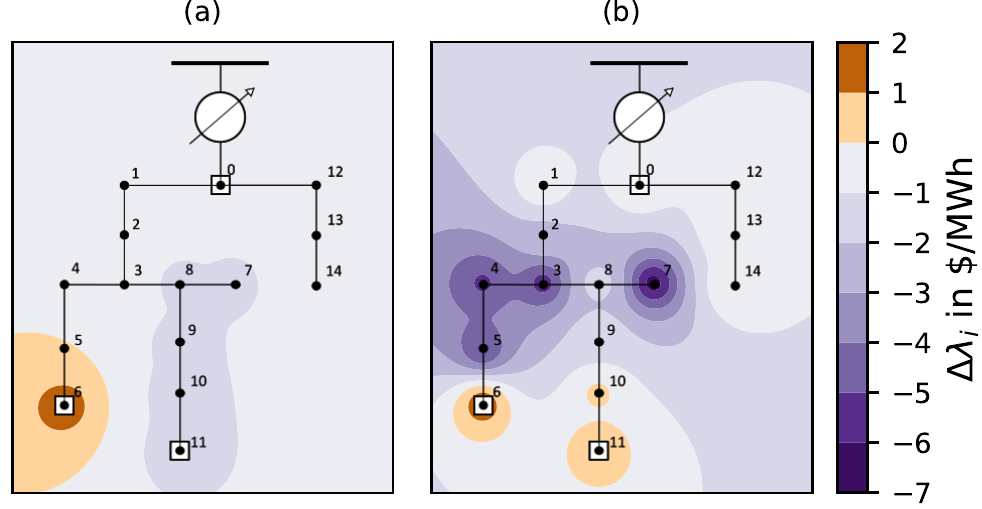}
\caption{DLMP difference $\Delta\lambda^P_i$ of (a) GEN-CC  and (b) VOLT-CC relative to the deterministic case.}
\label{fig:lambda_diffs_hm}
\end{figure}

The case study is performed on the 15-node radial feeder from \cite{papavasiliou2018analysis} with two minor modifications: one controllable DER is added at node 11 (see Fig.~\ref{fig:lambda_diffs_hm}) and the power flow limit of edges 2 and 3 is doubled to avoid congestion in the deterministic case.
Cost parameters of DERs at nodes 6 and 11 are set to $c_{1,i} = \unitfrac[10]{\$}{MWh}$, $c_{2,i} = \unitfrac[5]{\$}{MWh^2}$, $c_{0,i} = 0$.
The substation cost is set to  $c_{1,0} = \unitfrac[50]{\$}{MWh}$, $c_{2,0} = \unitfrac[400]{\$}{MWh^2}$, $c_{0,0} = 0$.
Note that this selection incentivizes the use of DERs.
The data of \cite{papavasiliou2018analysis} is used as scheduled net demand with a normally distributed zero-mean error, standard deviation of $\sigma_i^P = 0.2 d_i^P$ and no covariance among the nodes. 
The security parameter of the chance constraints is set to $\epsilon_g = \unit[5]{\%}$ and $\epsilon_v = \unit[1]{\%}$.
All models in the case study are implemented using the Julia JuMP package and our code can be downloaded from~\cite{drlearning_code}.

\subsection{Effect of uncertainty on DLMPs}

\Cref{tab:det_optimal_solution,tab:genCC_optimal_solution,tab:voltCC_optimal_solution} summarize the optimal solution and prices in the deterministic, GEN-CC and VOLT-CC cases.
Note that the deterministic case is solved for the expected net demand and $\alpha_i = 0, \forall i$.
In the deterministic and GEN-CC cases, none of the generator limits are active and, therefore, their power production does not differ. 
Similarly the resulting voltage magnitudes do not change as the GEN-CC  considers deterministic voltage constraints and only the flow limit of edges $8$ and $6$ are binding. 
In the VOLT-CC, however, the resulting voltage magnitudes are closer to unity in order to accommodate real-time power imbalances. 
As a result, the voltage constraints \cref{det-cc-opf:zeta_i,det-cc-opf:nu_i,det-cc-opf:mu_i+,det-cc-opf:mu_i-} in the VOLT-CC yield non-zero dual multipliers.
Figure~\ref{fig:lambda_diffs_hm} itemizes the effect of uncertainty on $\lambda^P_i$ relative to the deterministic case, where $\Delta\lambda^P_i = \lambda_i^{P,(\text{GEN-CC/VOLT-CC})} - \lambda_i^{P,(\text{DET})}$. 
While the passive branch of the system (nodes $12$ to $14$ without any controllable DERs) shows no changes in DLMPs as it is fully supplied by the substation, DLMPs vary  in the branches with DERs.

\begin{table}[t]
\renewcommand{\arraystretch}{0.8}
\setlength\tabcolsep{4 pt}
\caption{Optimal Deterministic Solution}
\label{tab:det_optimal_solution}
\centering
\begin{tabular}{c|cccccc}
\toprule
 $i$ &  $g_i^P$ & $g_i^Q$ & $\alpha_i$ &  $\sqrt{(f_i^{P})^2 + (f_i^{Q})^2}$ & $v_i$ & $\lambda^P_i$ \\
\midrule
   0 & 0.994 & 0.344 &    -- & 0.000 &    1.000 &  50.000 \\
   1 &    -- &    -- &    -- & 0.404 &    0.975 &  50.000 \\
   2 &    -- &    -- &    -- & 0.446 &    1.012 &  50.000 \\
   3 &    -- &    -- &    -- & 0.446 &    1.067 &  50.000 \\
   4 &    -- &    -- &    -- & 0.210 &    1.071 &  50.000 \\
   5 &    -- &    -- &    -- & 0.227 &    1.074 &  50.000 \\
   6 & 0.278 & 0.006 &    -- & 0.256\textsuperscript{*} &    1.086 &  10.411 \\
   7 &    -- &    -- &    -- & 0.197 &    1.086 &  10.208 \\
   8 &    -- &    -- &    -- & 0.256\textsuperscript{*} &    1.077 &  10.208 \\
   9 &    -- &    -- &    -- & 0.083 &    1.078 &  10.208 \\
  10 &    -- &    -- &    -- & 0.108 &    1.081 &  10.208 \\
  11 & 0.140 & 0.032 &    -- & 0.130 &    1.082 &  10.208 \\
  12 &    -- &    -- &    -- & 0.660 &    0.983 &  50.000 \\
  13 &    -- &    -- &    -- & 0.025 &    0.978 &  50.000 \\
  14 &    -- &    -- &    -- & 0.024 &    0.975 &  50.000 \\
\midrule
    & \multicolumn{6}{l}{\textsuperscript{*} Constraint is binding} \\
\bottomrule
\end{tabular}
\end{table}

\begin{table}
\renewcommand{\arraystretch}{0.8}
\setlength\tabcolsep{3.5 pt}
\caption{Optimal GEN-CC Solution}
\label{tab:genCC_optimal_solution}
\centering
\begin{tabular}{c|ccccccc}
\toprule
 $i$ &  $g_i^P$ & $g_i^Q$ &  $\alpha_i$ & $\sqrt{(f_i^{P})^2 + (f_i^{Q})^2}$ & $v_i$ & $\lambda^P_i$ & $\gamma$\\
\midrule
   0 & 0.994 & 0.344 & 0.003 & 0.000 &    1.000 &  50.00 & 0.273 \\
   1 &    -- &    -- &    -- & 0.404 &    0.975 &  50.00 & -- \\
   2 &    -- &    -- &    -- & 0.446 &    1.012 &  50.00 & -- \\
   3 &    -- &    -- &    -- & 0.446 &    1.067 &  50.00 & -- \\
   4 &    -- &    -- &    -- & 0.210 &    1.071 &  50.00 & -- \\
   5 &    -- &    -- &    -- & 0.227 &    1.074 &  50.00 & -- \\
   6 & 0.278 & 0.006 & 0.646 & 0.256\textsuperscript{*} &    1.086 &  11.99 & -- \\
   7 &    -- &    -- &    -- & 0.197 &    1.086 &   8.769 & -- \\
   8 &    -- &    -- &    -- & 0.256\textsuperscript{*} &    1.077 &   8.769 & -- \\
   9 &    -- &    -- &    -- & 0.083 &    1.078 &   8.769 & -- \\
  10 &    -- &    -- &    -- & 0.108 &    1.081 &   8.769 & -- \\
  11 & 0.140\textsuperscript{*} & 0.032 & 0.351 & 0.130 &    1.082 &   8.769 & -- \\
  12 &    -- &    -- &    -- & 0.660 &    0.983 &  50.00 & -- \\
  13 &    -- &    -- &    -- & 0.025 &    0.978 &  50.00 & -- \\
  14 &    -- &    -- &    -- & 0.024 &    0.975 &  50.00 & -- \\
\midrule
 & \multicolumn{7}{l}{\textsuperscript{*} Constraint is binding} \\
\bottomrule
\end{tabular}
\end{table}

\begin{table}[t]
\renewcommand{\arraystretch}{0.8}
\setlength\tabcolsep{3.5 pt}
\caption{Optimal VOLT-CC Solution}
\label{tab:voltCC_optimal_solution}
\centering
\begin{tabular}{c|ccccccc}
\toprule
 $i$ &  $g_i^P$ & $g_i^Q$ &  $\alpha_i$ & $\sqrt{(f_i^{P})^2 + (f_i^{Q})^2}$ & $v_i$ & $\lambda^P_i$ & $\gamma$ \\
\midrule
   0 & 1.033 &  0.490 & 0.377 & 0.000 &    1.000 &  50.00 & 31.51 \\
   1 &    -- &     -- &    -- & 0.523 &    0.956 &  49.97 & -- \\
   2 &    -- &     -- &    -- & 0.439 &    0.972 &  47.88 & -- \\
   3 &    -- &     -- &    -- & 0.439 &    0.996 &  44.59 & -- \\
   4 &    -- &     -- &    -- & 0.208 &    0.997 &  44.83 & -- \\
   5 &    -- &     -- &    -- & 0.222 &    0.999\textsuperscript{*} &  45.056 & -- \\
   6 & 0.258 & -0.068 & 0.370 & 0.256\textsuperscript{*} &    1.005\textsuperscript{*} &  12.03 & -- \\
   7 &    -- &     -- &    -- & 0.197 &    1.011\textsuperscript{*} &   3.884 & -- \\
   8 &    -- &     -- &    -- & 0.256\textsuperscript{*} &    1.001 &   8.577 & -- \\
   9 &    -- &     -- &    -- & 0.090 &    1.001 &   9.111 & -- \\
  10 &    -- &     -- &    -- & 0.100 &    1.001\textsuperscript{*} &  10.39 & -- \\
  11 & 0.121 & -0.040 & 0.253 & 0.116 &    1.002\textsuperscript{*} &  10.95 & -- \\
  12 &    -- &     -- &    -- & 0.660 &    0.983 &  50.00 & -- \\
  13 &    -- &     -- &    -- & 0.025 &    0.978 &  50.00 & -- \\
  14 &    -- &     -- &    -- & 0.024 &    0.975 &  50.00 & -- \\
\midrule
 & \multicolumn{7}{l}{\textsuperscript{*} Constraint is binding} \\
\bottomrule
\end{tabular}
\end{table}

\subsection{Price Decomposition}
Tables~\ref{tab:GEN-CC_decomp} and \ref{tab:VOLT-CC_decomp} itemize the components of the energy price following Proposition~\ref{prop:price_decomp}.
{\color{black} 
Additionally, Fig.~\ref{fig:constraint_limit_illustration} illustrates the nodes and edges with binding limits and, thus, non-zero Lagrangian multipliers.
}
Since the GEN-CC has no active voltage constraints the energy price at each node is determined by $\lambda^P_{\set{A}_i}$, i.e. the energy price at the ancestor node, and the congestion as per \cref{eq:active_power_price}.
There is no reactive power price component due to inactive voltage limits.
In the VOLT-CC case, on the other hand, the voltage limits become active and therefore reactive power price is non-zero.

In Fig.~\ref{fig:lambda_diffs_hm} we observe higher prices at and close to the nodes with DERs.
As follows from Eq.~\cref{eq:lambda_i_decomp} and summarized in Table~\ref{tab:VOLT-CC_decomp_mus}, prices at those nodes are dominated by the lower voltage limits, thus quantifying the value of downward regulation.
A negative net demand value at node $7$ indicates a high uncontrolled behind-the-meter generation, which leads to low prices dominated by the upper voltage limit. 
This incentivizes  a higher demand and lower generation.
{\color{black} 
At node $6$ both the upper and lower voltage limits are binding (see Fig.~\ref{fig:constraint_limit_illustration} and non-zero $\mu_6^+, \mu_6^-$ in the bottom row of Table~\ref{tab:VOLT-CC_decomp_mus}), thus indicating that no more balancing regulation at this node is possible without increasing the likelihood of voltage limit violations. 
Hence, the trade-off between the power output and balancing regulation of the DER at node $6$ is no longer driven by its profit-maximizing objective, but rather by the the physical limits of the system.
}

\begin{figure}
    \centering
    \includegraphics[width=1\linewidth]{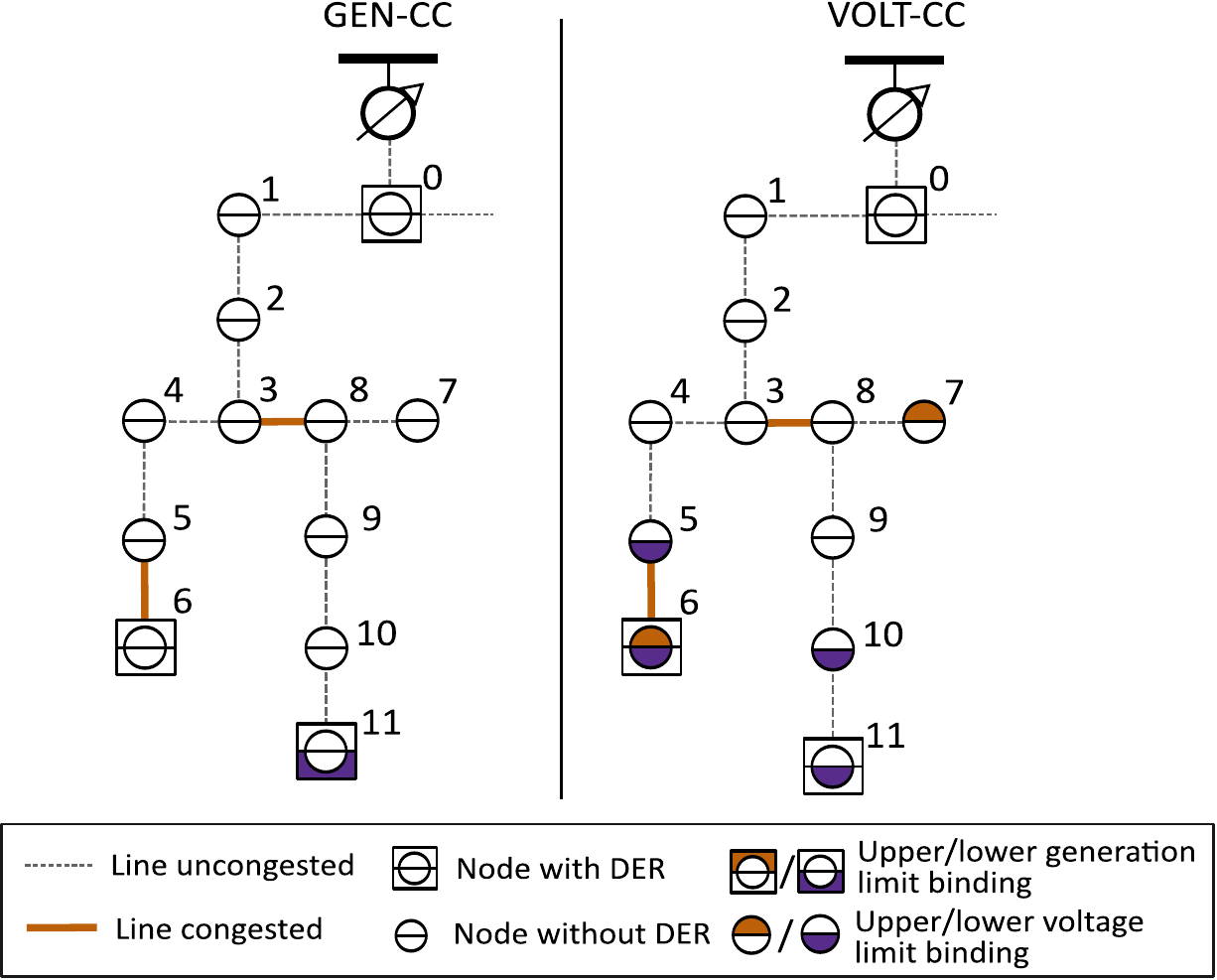}
    \caption{Illustration of binding edge, voltage and generation limits in the GEN-CC and the VOLT-CC.}
    \label{fig:constraint_limit_illustration}
\end{figure}

\begin{table}[t]
\renewcommand{\arraystretch}{0.8}
\setlength\tabcolsep{4pt}
\caption{DLMP Decomposition of the GEN-CC, cf. Eq.~\eqref{eq:active_power_price}}
\label{tab:GEN-CC_decomp}
\centering
\begin{tabular}{c|c|cccc}
\toprule
 $i$ &  $\lambda^P_i$ & $\lambda^P_{\set{A}_i}$ & $\lambda^Q_i\frac{r_i}{x_i}$ & $\lambda^Q_{\set{A}_i}\frac{r_i}{x_i}$ & $2\eta_i(f_i^P\!+\!\frac{r_i}{x_i} f_i^Q)$ \\
\midrule
   0 &  50.000 &      -0.000 &   -0.000 &   -0.000 &  0.000 \\
   1 &  50.000 &      50.000 &    0.000 &    0.000 & -0.000 \\
   2 &  50.000 &      50.000 &    0.000 &    0.000 &  0.000 \\
   3 &  50.000 &      50.000 &    0.000 &    0.000 &  0.000 \\
   4 &  50.000 &      50.000 &    0.000 &    0.000 & -0.000 \\
   5 &  50.000 &      50.000 &    0.000 &    0.000 &  0.000 \\
   6 &  11.996 &      50.000 &    0.000 &    0.000 & 38.004 \\
   7 &   8.769 &       8.769 &    0.000 &    0.000 &  0.000 \\
   8 &   8.769 &      50.000 &    0.000 &    0.000 & 41.231 \\
   9 &   8.769 &       8.769 &    0.000 &    0.000 & -0.000 \\
  10 &   8.769 &       8.769 &    0.000 &    0.000 &  0.000 \\
  11 &   8.769 &       8.769 &    0.000 &    0.000 & -0.000 \\
  12 &  50.000 &      50.000 &    0.000 &    0.000 &  0.000 \\
  13 &  50.000 &      50.000 &    0.000 &    0.000 &  0.000 \\
  14 &  50.000 &      50.000 &    0.000 &    0.000 &  0.000 \\
\bottomrule
\end{tabular}
\end{table}

\begin{table}[t]
\renewcommand{\arraystretch}{0.8}
\setlength\tabcolsep{4pt}
\caption{DLMP Decomposition of the VOLT-CC, cf. Eq.~\eqref{eq:active_power_price}}
\label{tab:VOLT-CC_decomp}
\centering
\begin{tabular}{c|c|cccc}
\toprule
 $i$ &  $\lambda^P_i$ & $\lambda^P_{\set{A}_i}$ & $\lambda^Q_i\frac{r_i}{x_i}$ & $\lambda^Q_{\set{A}_i}\frac{r_i}{x_i}$ & $2\eta_i(f_i^P\!+\!\frac{r_i}{x_i} f_i^Q)$ \\
\midrule
   0 &  50.000 &      -0.000 &   -0.000 &   -0.000 &  0.000 \\
   1 &  49.976 &      50.000 &    0.024 &    0.000 & -0.000 \\
   2 &  47.881 &      49.976 &    4.088 &    1.993 &  0.000 \\
   3 &  44.596 &      47.881 &    7.373 &    4.088 &  0.000 \\
   4 &  44.836 &      44.596 &    7.132 &    7.373 & -0.000 \\
   5 &  45.056 &      44.836 &    6.887 &    7.108 &  0.000 \\
   6 &  12.025 &      45.056 &    0.000 &    6.911 & 39.942 \\
   7 &   3.884 &       8.577 &    7.072 &    2.379 &  0.000 \\
   8 &   8.577 &      44.596 &    2.376 &    7.369 & 41.012 \\
   9 &   9.111 &       8.577 &    1.842 &    2.376 & -0.000 \\
  10 &  10.399 &       9.111 &    0.552 &    1.840 &  0.000 \\
  11 &  10.949 &      10.399 &    0.000 &    0.550 & -0.000 \\
  12 &  50.000 &      50.000 &    0.000 &    0.000 &  0.000 \\
  13 &  50.000 &      50.000 &    0.000 &    0.000 &  0.000 \\
  14 &  50.000 &      50.000 &    0.000 &    0.000 &  0.000 \\
\bottomrule
\end{tabular}
\end{table}

\begin{table}[t]
\renewcommand{\arraystretch}{0.8}
\setlength\tabcolsep{4pt}
\caption{DLMP Decomposition of the VOLT-CC based on voltage constraints, cf. Eq.~\eqref{eq:lambda_i_decomp}}
\label{tab:VOLT-CC_decomp_mus}
\centering
\begin{tabular}{c|c|c|ccc}
\toprule
 $i$ &  $\lambda^P_i$ & $(\mu_i^+\!-\!\mu_i^-)$ & $\lambda^P_{\set{A}_i}$ & $2r_i\sum_{j\in\set{D}_i}(\mu_j^+\!-\!\mu_j^-)$ & $2f_i^P \eta_i$ \\
\midrule
   0 &  50.000 &   0.000 &      -0.000 &      -0.000 &   0.000 \\
   1 &  49.976 &  -0.000 &      50.000 &       0.024 &  -0.000 \\
   2 &  47.881 &  -0.000 &      49.976 &       2.096 &   0.000 \\
   3 &  44.596 &  -0.000 &      47.881 &       3.285 &   0.000 \\
   4 &  44.836 &  -0.000 &      44.596 &      -0.240 &  -0.000 \\
   5 &  45.056 &  -0.001 &      44.836 &      -0.220 &   0.000 \\
   6 &  12.025 &  -6.290\textsuperscript{\textdagger} &      45.056 &      -0.606 &  33.638 \\
   7 &   3.884 &  44.871 &       8.577 &       4.694 &   0.000 \\
   8 &   8.577 &  -0.000 &      44.596 &       1.478 &  34.541 \\
   9 &   9.111 &  -0.000 &       8.577 &      -0.534 &  -0.000 \\
  10 &  10.399 &  -0.004 &       9.111 &      -1.288 &   0.000 \\
  11 &  10.949 & -26.709 &      10.399 &      -0.550 &  -0.000 \\
  12 &  50.000 &  -0.000 &      50.000 &      -0.000 &   0.000 \\
  13 &  50.000 &  -0.000 &      50.000 &      -0.000 &   0.000 \\
  14 &  50.000 &  -0.000 &      50.000 &      -0.000 &   0.000 \\
\midrule
 & \multicolumn{5}{l}{\textsuperscript{\textdagger} $\mu_6^+ = 20.186$, $\mu_6^- = 26.476$}    \\
\bottomrule
\end{tabular}
\end{table}

The regulation price in the GEN-CC case ($\gamma^{\text{GEN-CC}} = 0.273$) is notably lower relative to the VOLT-CC case ($\gamma^{\text{VOLT-CC}} = 31.511$).
As per Proposition~\ref{prop:gamma}, $\gamma^{\text{GEN-CC}}$ is only driven by the power output limits (Table~\ref{tab:gamma_decomp}), where only the lower output limit at node $11$ is binding.
Due to a  low power price at node $11$ as compared to node $6$, the scheduled power production is also low, which limits the downward regulation capacity provided.
By introducing voltage chance constraints in the VOLT-CC case, the DLMP composition changes as per Proposition~\ref{prop:gamma_with_voltCC} (Table~\ref{tab:gamma_decomp}).
Each node with binding voltage constraints ($5,6,7,10,11$) contributes to the formation of $\gamma$ by weighting the impact of the system-wide regulation participation on the voltage standard deviation against the marginal value of relaxed voltage limits for each node (Eq.~\ref{eq:nu_decomp}).

\begin{table}[t]
\renewcommand{\arraystretch}{0.8}
\caption{Regulation Price Decomposition of the VOLT-CC}
\label{tab:gamma_decomp}
\centering
\begin{tabular}{c c|cc|ccc}
\toprule
& & \multicolumn{2}{c|}{GEN-CC} & \multicolumn{3}{c}{VOLT-CC} \\
 $i$ &  $b_i$ & $\delta_i^+$ & $\delta_i^-$ &  $\delta_i^+$ & $\delta_i^-$ & $\nu_i$  \\
\midrule
  0 & 0.0005 &         -- &          -- &     -- &          -- & -- \\
   1 &    -- &         -- &          -- &     -- &          -- &  0.120 \\
   2 &    -- &         -- &          -- &     -- &          -- & 10.699 \\
   3 &    -- &         -- &          -- &     -- &          -- & 27.281 \\
   4 &    -- &         -- &          -- &     -- &          -- & 28.199 \\
   5 &    -- &         -- &          -- &     -- &          -- & 29.040 \\
   6 &   0.1 &     -0.000 &      -0.000 & -0.000 &      -0.000 & 31.357 \\
   7 &    -- &         -- &          -- &     -- &          -- & 32.533 \\
   8 &    -- &         -- &          -- &     -- &          -- & 30.201 \\
   9 &    -- &         -- &          -- &     -- &          -- & 30.472 \\
  10 &    -- &         -- &          -- &     -- &          -- & 31.126 \\
  11 &   0.1 &     -0.000 &       0.310 & -0.000 &      -0.000 & 31.406 \\
  12 &    -- &         -- &          -- &     -- &          -- &  0.000 \\
  13 &    -- &         -- &          -- &     -- &          -- &  0.000 \\
  14 &    -- &         -- &          -- &     -- &          -- &  0.000 \\
\bottomrule
\end{tabular}
\end{table}

{\color{black} 
\subsection{Impact of Losses}

\begin{table}[t]
\renewcommand{\arraystretch}{0.8}
\setlength\tabcolsep{3.5 pt}
\caption{Optimal LVOLT-CC Solution}
\label{tab:losses_voltCC_optimal_solution}
\centering
{\color{black}
\begin{tabular}{c|ccccccc}
\toprule
 $i$ &  $g_i^P$ & $g_i^Q$ &  $\alpha_i$ & $\sqrt{(f_i^{P})^2 + (f_i^{Q})^2}$ & $v_i$ & $\lambda^P_i$ & $\gamma$ \\
\midrule
   0 & 1.075 &  0.607 & 0.144 & 0.000 &    1.000 &  50.000 & 30.254 \\
   1 &    -- &     -- &    -- & 0.576 &    0.952 &  49.971 &     -- \\
   2 &    -- &     -- &    -- & 0.422 &    0.966 &  47.417 &     -- \\
   3 &    -- &     -- &    -- & 0.433 &    0.997 &  43.414 &     -- \\
   4 &    -- &     -- &    -- & 0.209 &    0.998 &  43.666 &     -- \\
   5 &    -- &     -- &    -- & 0.223 &    0.999 &  43.897 &     -- \\
   6 & 0.256 & -0.077 & 0.515 & 0.256\textsuperscript{*} &    1.005\textsuperscript{*} &  14.005 &     -- \\
   7 &    -- &     -- &    -- & 0.197 &    1.016\textsuperscript{*} &   6.372 &     -- \\
   8 &    -- &     -- &    -- & 0.256\textsuperscript{*} &    1.006 &   4.169 &     -- \\
   9 &    -- &     -- &    -- & 0.079 &    1.007 &   4.169 &     -- \\
  10 &    -- &     -- &    -- & 0.102 &    1.009 &   4.169 &     -- \\
  11 & 0.137 &  0.012 & 0.341 & 0.124 &    1.011 &   4.169 &     -- \\
  12 &    -- &     -- &    -- & 0.661 &    0.983 &  50.000 &     -- \\
  13 &    -- &     -- &    -- & 0.026 &    0.978 &  50.000 &     -- \\
  14 &    -- &     -- &    -- & 0.024 &    0.975 &  50.000 &     -- \\
\midrule
 & \multicolumn{7}{l}{\textsuperscript{*} Constraint is binding} \\
\bottomrule
\end{tabular}
}
\end{table}

\begin{figure}
    \centering
    \includegraphics[width=\linewidth]{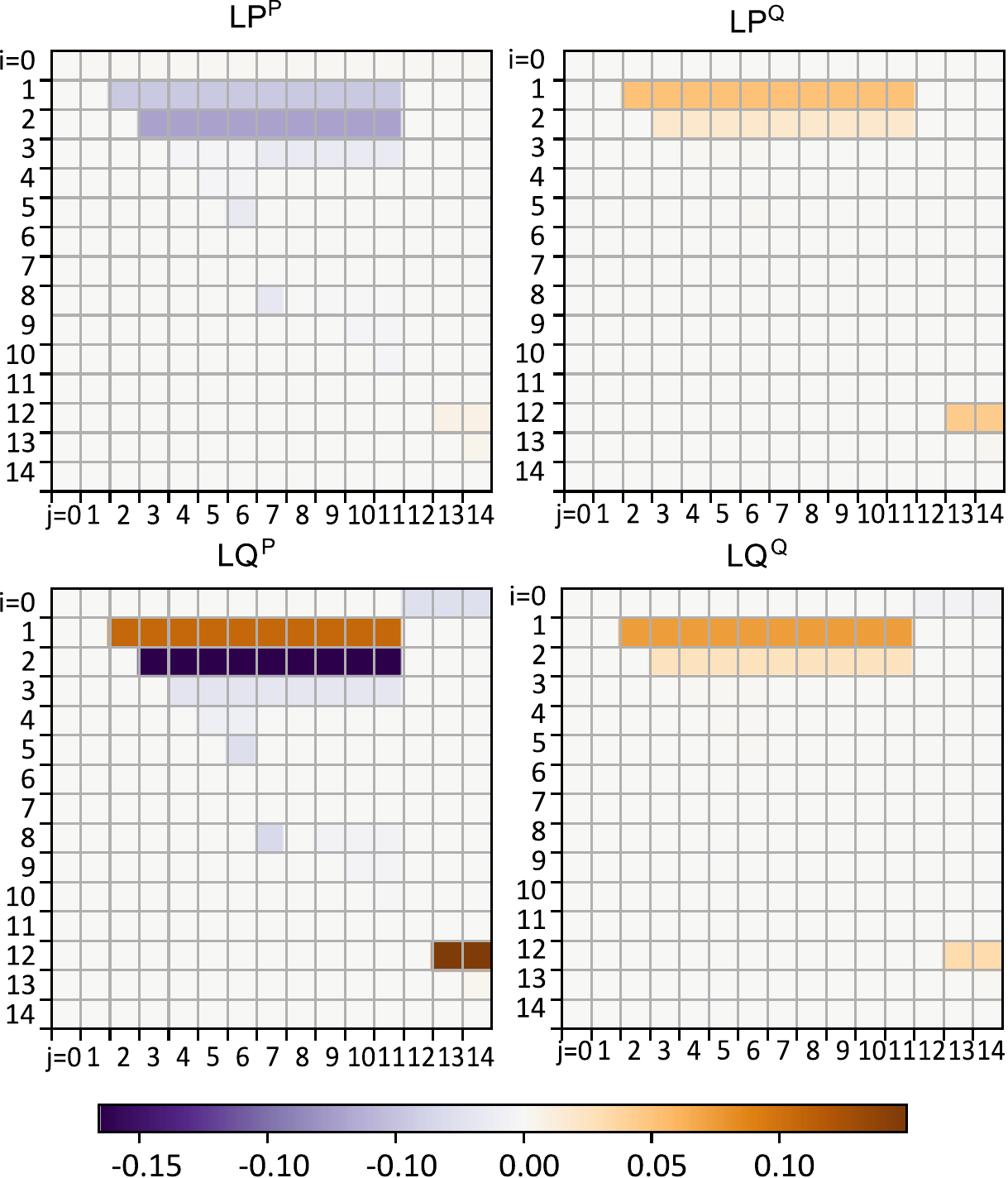}
    \caption{Illustration of loss factor matrices $\LP^P$, $\LP^Q$, $\LQ^P$, $\LQ^Q$ itemizing the sensitivity of active and reactive nodal net injections at node $j$ (`x-axis') on the active and reactive FND of node $i$ (`y-axis').}
    \label{fig:loss_factor_matrix}
\end{figure}

Table~\ref{tab:losses_voltCC_optimal_solution} summarizes the optimal LVOLT-CC solution obtained by using loss factor matrices $\LP^{P}$, $\LP^{Q}$, $\LQ^{P}$, $\LQ^{Q}$ as defined in \cref{eq:LP_entries,eq:LQ_entries} and shown in Fig.~\ref{fig:loss_factor_matrix}.
Negative elements of the matrices shown in Fig.~\ref{fig:loss_factor_matrix} indicate that additional DER production at node $j$ will increase power losses allocated to node $i$ based on FND.
For example, $\LP^P$ shows that additional active production at nodes $0$ to $11$ will increase active power losses. 
On the other hand, additional active production on the passive branch (nodes $12$ to $14$), where no DERs are installed, will reduce active power losses. 

Since increasing DER production at nodes $6$ and $11$ increases system losses, see  Fig.~\ref{fig:loss_factor_matrix}, we observe that the power output of controllable DERs changes slightly, as compared to the results of the VOLT-CC, and  the additional power  needed  to compensate for system  losses is  provided by the substation (node $0$).
By internalizing the loss factors into the voltage chance constraints via matrix $R^L$ as in \cref{eq:R_lossaware}, the impact of balancing participation on the voltage limits is no longer symmetric as in~\cref{eq:def_R_X}.
Thus, we observe higher balancing participation factors of the DERs at nodes $6$ and $11$ relative to the VOLT-CC.
Additionally, a greater power supply from the substation and the DER at node $11$ leads to non-binding voltage constraints and, thus, uniform DLMPs at nodes $8$~to~$11$.
In line with the theoretical results of Proposition~\ref{prop:including_losses}a), the additional power losses have almost no impact on the price for balancing regulation.
The small difference relative to the VOLT-CC is mainly caused by non-binding voltage constraints at nodes $10$~and~$11$.
}

\section{Conclusion}

{\color{black}
This paper described an approach to derive stochasticity-aware DLMPs for electricity pricing in low-voltage electric power distribution systems that explicitly internalize uncertainty and risk parameters.} 
These DLMPs are also shown to constitute a robust competitive equilibrium, which can be leveraged towards emerging distribution electricity  market designs. In the future, our work will focus on the application of the proposed pricing theory to decentralized and communication-constrained control of DERs and for enabling electricity pricing in distribution systems with a high penetration rate of DERs and near-zero marginal production costs. 
{\color{black} 
Methodological extensions can encompass uncertainty internalization via semidefinite programming to allow for non-linear power flow representations, \cite{wang2017accurate}, 
{\color{black} and the impact of asymmetric information and strategic behavior.}
}

\bibliographystyle{IEEEtran}
\bibliography{literature_local}

\appendix
\numberwithin{equation}{subsection}

\section{Deterministic Reformulations}

\subsection{Expected Generation Cost}
\label{ax:expected_generation_cost}
\begin{equation}
    \begin{split}
    & c_i(\bm{g}_i^P)  = \frac{(g_i^P + \alpha_i \Om + a_i)^2}{2 b_i} + c_i \\
    & = \frac{(g_i^P)^2 + \alpha_i^2 \Om^2 + a_i^2 +2(g_i^P \alpha_i \Om + g_i^P a_i + \alpha_i \Om a_i)}{2 b_i} + c_i.    
    \end{split}
\end{equation}
Recall that $\Eptn[\om] = 0$. Then $\Eptn[e^{\!\top}\om] = \Eptn[\Om] = 0$ and we obtain:
\allowdisplaybreaks
\begin{equation}
    \begin{split}
    \Eptn[c_i(\bm{g}_i^P)] 
    & = \frac{(g_i^P + a_i)^2}{2 b_i} + c_i + \Eptn\left[\frac{\alpha_i^2 \Om^2}{2b_i}\right]  \\
    & = c_i(g_i^P) + \frac{\alpha_i^2}{2b_i}\Eptn\left[\Om^2\right]        
    \end{split}
\end{equation}
\allowdisplaybreaks[0]

As $\Var[\bm{\Om}] = \Eptn[\Om^2] - \Eptn[\Om]^2$ and $\Var[\Om] = \Var[e^{\!\top}\om] = e^{\!\top}\Sigma e = s^2$ we find the respective expression $\Eptn\left[\Om^2\right]$.

\subsection{Voltage Variance}
\label{ax:voltage_variance}

Using the expression for the uncertain voltage \cref{eq:voltage_reformulated}, we can obtain 
\allowdisplaybreaks
\begin{align}
        \bm{u}_i(\om, \alpha) - \bar{u}_i 
        &= 2 [R (\om -\alpha e^{\!\top} \om]_i \nonumber\\
        &= 2 [R \om]_i - [R\alpha]_i (e^{\!\top} \om) \\
        &= 2 (R_{i*} - \rho_i^v e^{\!\top}) \om   . \nonumber
\end{align}
\allowdisplaybreaks[0]
Then Eq.~\cref{eq:voltage_stdev_conic} follows from:
\begin{align}
  \Var[\bm{u}_i(\om,\alpha)] = 4 \norm{(R_{i*} - \rho_i^v e^{\!\top})\Sigma^{1/2}}_2^2.
\end{align}

\end{document}